\documentclass[12pt]{article}

\usepackage{hyperref}

\usepackage{amsfonts}

\usepackage{amscd}
\usepackage{amsthm}
\usepackage{indentfirst}

\usepackage{amssymb, amsmath, amsthm, amsgen, amstext, amsbsy, amsopn}
\usepackage{epsfig}

\usepackage{color}
\usepackage{enumerate}
\usepackage{enumitem}

\usepackage{geometry}
\usepackage{a4}
\usepackage[T1]{fontenc}
\usepackage[utf8]{inputenc}
\usepackage{lmodern}
\usepackage{array}
\usepackage{latexsym}
\usepackage{fancyhdr}
\usepackage{mathrsfs}\let\cal\mathscr
\usepackage[all]{xy}
\usepackage{makeidx}
\usepackage{color}
\usepackage{pifont}
\usepackage{amsfonts,amssymb}
\usepackage{cleveref}
\usepackage[numbers,square]{natbib}
\usepackage{verbatim}
\usepackage{relsize}

\usepackage{avant}
\usepackage[T1]{fontenc}

\usepackage{graphicx}

\newtheorem{thm}{Theorem}[section]
\newtheorem{lemma}[thm]{Lemma}
\newtheorem{prop}[thm]{Proposition}
\newtheorem{cor}[thm]{Corollary}
\newtheorem{defin}[thm]{Definition}
\newtheorem{rem}[thm]{Remark}
\newtheorem{exam}[thm]{Example}
\newtheorem{problem}[thm]{Problem}

\newcommand{\PPP}{{\mathbb{P}}}
\newcommand{\R}{{\mathbb{R}}}

\newcommand{\Z}{{\mathbb{Z}}}
\newcommand{\N}{{\mathbb{N}}}
\newcommand{\C}{{\mathbb{C}}}

\newcommand{\cA}{{\mathcal{A}}}
\newcommand{\cB}{{\mathcal{B}}}
\newcommand{\cD}{{\mathcal{D}}}
\newcommand{\cC}{{\mathcal{C}}}
\newcommand{\cE}{{\mathcal{E}}}

\newcommand{\cH}{{\mathcal{H}}}

\newcommand{\cL}{{\mathcal{L}}}

\newcommand{\cN}{{\mathcal{N}}}
\newcommand{\cP}{{\mathcal{P}}}

\newcommand{\cS}{{\mathcal{S}}}

\newcommand{\cV}{{\mathcal{V}}}

\def\id{{1\hskip-2.5pt{\rm l}}}

\newcommand{\supp}{{\it supp\,}}

\newcommand{\om}{{\omega}}
\newcommand{\Om}{\Omega}

\newcommand{\bigo}{\mathcal{O}}
\newcommand{\Hilb}{\mathcal{H}}

\newcommand{\tr}{{{\rm tr}}}

\renewcommand \leq {\leqslant}
\renewcommand \geq {\geqslant}

\renewcommand \tilde[1] {\stackrel{\sim}{\smash{#1}\rule{0pt}{1.1ex}}}

\DeclareMathOperator{\Vol}{Vol}
\DeclareMathOperator{\End}{End}

\DeclareMathOperator{\Tr}{Tr}

\DeclareMathOperator{\Ker}{Ker}

\DeclareMathOperator{\Spec}{Spec}

\DeclareMathOperator{\diag}{diag}

\newcommand \< {\mathcal{h}}
\renewcommand \> {\mathcal{i}}

\newcommand \cinf {\CC^\infty}

\renewcommand \epsilon {\varepsilon}

\newcommand \CC {{\cal C}}

\newcommand \HH {{\cal H}}

\newcommand \TT {{\cal T}}

 \def\cC{\mathscr{C}}
 \def\cD{\mathscr{D}}
\def\cL{\mathscr{L}}

\def\Im{{\rm Im}}

\newcommand{\Prod}{{\mathcal Prod}}

\newcommand \fl {\rightarrow}

\newcommand \ignore[1] {}

\crefname{equation}{}{}
\crefname{lemma}{Lemma}{Lemmas}
\crefname{thm}{Theorem}{Theorems}
\crefname{cor}{Corollary}{Corollaries}
\crefname{ex}{Example}{Examples}
\crefname{defi}{Definition}{Definitions}
\crefname{prop}{Proposition}{Propositions}
\crefname{section}{Section}{Sections}
\crefname{subsection}{Section}{Sections}
\crefname{rmk}{Remark}{Remarks}
\crefname{nota}{Notation}{Notations}

\begin{document}

\title{Spectral aspects of the Berezin transform}

\renewcommand{\thefootnote}{\alph{footnote}}

\author{\textsc Louis Ioos$^{a}$, Victoria Kaminker$^{b}$, Leonid Polterovich$^{b}$ and Dor Shmoish$^{b}$ }

\footnotetext[1]{Partially supported by the European Research Council Starting grant 757585}
\footnotetext[2]{Partially supported by the European Research Council Advanced grant 338809 }

\date{\today}

\maketitle

\begin{abstract}
We discuss the Berezin transform, a Markov operator associated
to positive operator valued measures (POVMs),
in a number of contexts including the Berezin-Toeplitz quantization,
Donaldson's dynamical system on the space of Hermitian inner products on
a complex vector space, representations of finite groups, and quantum noise.
In particular, we calculate the spectral gap for quantization in terms of the
fundamental tone of the phase space. Our results confirm a prediction of Donaldson for the spectrum of the $Q$-operator on K\"{a}hler manifolds with constant scalar curvature,
and yield exponential convergence of Donaldson's iterations to the fixed point.
Furthermore, viewing POVMs as data clouds, we study their spectral features via geometry of
measure metric spaces and the diffusion distance.
\end{abstract}

\tableofcontents

\section{Introduction} Given a function $f$ on a classical phase space $X$, let us first quantize it and then dequantize. This operation on functions, $f \mapsto \cB f$, is called {\it the Berezin transform}. As a result of this operation, the function $f$ blurs on the phase space. The intuition behind this is as follows \footnote{We thank S.~Nonnenmacher for this explanation.}: assume that $f$ is the Dirac delta-function at a point $x \in X$. Its quantization is a coherent state at $x$, whose dequantization is approximately a Gaussian centered at $x$. In the framework of the Berezin-Toeplitz quantization of closed K\"{a}hler
manifolds, $\cB$ is known to be a Markov operator with finite-dimensional image, and is closely related to the Laplace-Beltrami operator $\Delta$ of the K\"ahler manifold. In fact, the Berezin transform has the following asymptotic expansion as $\hbar\to 0$, due to Karabegov and Schlichenmaier \cite{KS01}
\footnote{Note that after renormalization, there is a missing
factor of $1/2$ in front of
the second term of the analogous formula in \cite[(1.2)]{KS01}.}:
\begin{equation}\label{eq-asymp} \cB_\hbar(f)= f-\frac{\hbar}{4\pi}\Delta f + \bigo(\hbar^2)\;,
\end{equation}
for every smooth function $f$ on $X$, with remainder depending on
$f$ and where $\hbar$ stands for the Planck constant
(see \cref{subsec-BT} for notations and conventions).

We focus on the spectral properties of $\cB$. For fixed $\hbar$,
this operator factors through a finite-dimensional space and hence its spectrum consists of a finite collection of points lying in the interval $[0,1]$. Moreover, multiplicities of positive eigenvalues are finite, and $1$ is the maximal eigenvalue corresponding to the constant function. Write its spectrum (with multiplicities) in the form $$1=\gamma_0
\geq \gamma_1 \geq \gamma_2\geq \dots \geq \gamma_k \geq \dots \geq 0\;.$$
The quantity $\gamma:= 1-\gamma_1$ is called {\it the spectral gap}, a fundamental characteristic of a Markov chain responsible for the rate of convergence to the stationary distribution. Our first result, Theorem \ref{thm-quant}, implies that in the context of the Berezin-Toeplitz quantization, the spectral gap $\gamma$ of the Berezin transform equals
\begin{equation}\label{eq-gap-intro}
\gamma = \frac{\hbar}{4\pi}\lambda_1 + \bigo(\hbar^2)\;,
\end{equation}
where $\lambda_1$ stands for the first eigenvalue of $\Delta$.
Note that the upper bound on the gap readily follows from
\eqref{eq-asymp}. The proof
follows a work of Lebeau and Michel \cite{LM} on semiclassical random walks on manifolds with extra ingredients such as an asymptotic expansion for the Bergman kernel due to Dai, Liu and Ma \cite{DLM06},
a comparison between the Berezin transform and the heat operator motivated by the work of Liu and Ma
\cite{LM07},
and a refined version of the above-mentioned Karabegov-Schlichenmaier asymptotic expansion \cite{KS01}. In fact, Theorem \ref{thm-quant} shows much more than \eqref{eq-gap-intro}, namely that one can approximate the
full spectrum of $\Delta$, as well as the associated eigenfunctions, with those of $\cB$.

Let us point out that the proof of Theorem \ref{thm-quant}
can be extended
to Berezin-Toeplitz quantization of closed
symplectic manifolds, using the quantum spaces given by the
eigenstates corresponding to the small eigenvalues of
the renormalized Bochner Laplacian. This uses the
associated generalized Bergman kernel of Ma and Marinescu
\cite{MM08} and asymptotic estimates refining those
of Ma, Marinescu and Lu \cite{LMM16}
(see the discussion at the end of Section \ref{subsec-BT}).

\medskip

The Berezin transform is defined in the more general context of positive operator valued measures (POVMs).
In fact, the Berezin-Toeplitz quantization is nothing else but the integration over a certain POVM on the phase space $M$
with values in the space of quantum observables, and the dequantization is the dual operation \cite{La,CP}. In addition to quantization, POVMs appear in quantum mechanics in another setting: they model quantum measurements \cite{Busch-2}. Interestingly enough, within this model the spectral gap of the Berezin transform corresponding to a POVM admits two different interpretations: it measures the minimal magnitude of quantum noise production, and it equals the spectral gap of the Markov chain corresponding to repeated quantum measurements (see Section \ref{sec-noise} for details).

\medskip

Another theme of this paper is related to Donaldson's program \cite{D} of developing approximate
methods for detecting canonical metrics on K\"{a}hler manifolds. Interestingly enough, our study
of the Berezin transform yields the asymptotic behaviour of the spectrum and of the eigenfunctions of the $Q$-operator, a geometric operator arising in this program, for   K\"ahler metrics of constant
scalar curvature. This behaviour, which  was predicted by Donaldson in \cite{D},
is stated in Theorem \ref{cor-quant} below.

Additionally, Donaldson discovered in \cite{D} a remarkable
class of dynamical systems on the space of all Hermitian products on a given complex vector space.
Section \ref{sec-don} deals with the  spectrum of the linearization of such a system  at a fixed point.
We show that it can be identified with the quantum channel associated
to a certain POVM. Using the positivity of the associated spectral
gap and under certain natural assumptions, we prove that
this linearization is contracting, which confirms Donaldson's
prediction via numerical computations in \cite[\S\,3]{D}.
By the Grobman-Hartman theorem
and earlier results of Donaldson,
this implies in particular
that the iterations of this system converge exponentially fast to its
fixed point (see Theorem \ref{thm-expcvcor}), and not only for "almost
all initial conditions", as predicted in \cite[\S\,4.1]{D}.
The use of Hartman's theorem in
a related context has been suggested by Fine in
\cite{Fin12}.

\medskip

This naturally brings us, in Section \ref{sec-bestfit}, to a geometric viewpoint at POVMs.
Following Oreshkov and Calsamiglia \cite[VII.C]{OC}, we encode them as probability measures in the space of quantum states $\cS$ equipped with the Hilbert-Schmidt metric. It turns out that the spectral gap admits a transparent description in terms of the geometry of such metric measure spaces and exhibits a robust behaviour under perturbations of POVMs in the Wasserstein metric. In a similar spirit, one can consider a POVM as a data cloud in $\cS$, which leads us to a link between the spectral gap and the diffusion distance, a notion coming from geometric data analysis.

Section \ref{sec-groups} contains a case study of POVMs associated to irreducible unitary representations of finite groups. In this case the spectrum of the Berezin transform and the diffusion distance associated to the corresponding Markov chain can be calculated explicitly via the character table of the group, and their properties reflect algebraic features. In particular, we prove that any non-trivial irreducible representation of a simple group has a strictly positive spectral gap (see Corollary \ref{cor-simplegr}).

\section{Preliminaries}\label{prel}
The mathematical model of quantum mechanics starts with a complex Hilbert space $\Hilb$. In what follows
we consider finite-dimensional Hilbert spaces only. Observables are represented by Hermitian operators whose space is denoted by $\cL(\Hilb)$.  Quantum states are provided by {\it density operators}, i.e., positive trace-one operators $\rho \in \cL(\Hilb)$. They form a subset $\cS(\Hilb) \subset \cL(\Hilb)$.
{\bf Notation:} We write $((A,B))$ for the scalar product $\tr(AB^*)= \tr(AB)$ on $\cL(\Hilb)$.

Let $\Omega$ be a set equipped with a $\sigma$-algebra $\cC$ of its subsets. By default,
we assume that $\Omega$ is a Polish topological space (i.e., it is homeomorphic to
a complete metric space possessing a countable dense subset) and $\cC$ is the Borel $\sigma$-algebra.

An $\cL(\Hilb)$-valued {\it positive operator valued measure} $W$
on $(\Omega,\cC)$, which we abbreviated to POVM,
is a countably additive map $W\colon \cC \to \cL(\Hilb)$ which takes each subset $X \in \cC$ to a positive operator $W(X) \in \cL(\Hilb)$ and which is normalized by $W(\Omega) = \id$.
According to \cite{CDS}, every $\cL(\Hilb)$-valued POVM possesses a  density with respect to some probability measure $\alpha$ on $(\Omega, \cC)$, that is having the form
\begin{equation}\label{eq-POVM-density}
dW(s) = nF(s) d\alpha(s)\;,
\end{equation}
where $n= \dim_{\C} \Hilb$ and $F: \Omega \to \cS(\Hilb)$ is a measurable function.

A POVM $W$ given by formula \eqref{eq-POVM-density} is called {\it pure} if
\begin{itemize}
\item[{(i)}] for every $s \in \Omega$  the state $F(s)$ is pure, i.e. a rank one projector;
\item[{(ii)}] the map $F:\Omega \to \cS(\Hilb)$ is one to one.
\end{itemize}
Pure POVMs, under various names, arise in several areas of mathematics including the Berezin-Toeplitz quantization , convex geometry (see \cite{GM} for the notion of an isotropic measure and \cite{AS} for the resolution of identity associated to John and L\"{o}wner ellipsoids), signal processing (see \cite{EF} for a link between tight frames and quantum measurements) and Hamiltonian group actions \cite{FM}. When $\Omega$ is a finite set, a pure POVM with a given measure $\alpha$ exists if and only if the measure $\alpha(\{s\})$ of each point $s \in \Omega$ is $\leq 1/n$, see \cite{FM} for a detailed account on the structure of the moduli spaces of pure POVMs on finite sets up to unitary conjugations.

Let us introduce the  main character of our story, the spectral gap of a POVM of the form \eqref{eq-POVM-density}.
Define a map $T: L_1(\Omega, \alpha) \to \cL(\Hilb)$ by
$$T(\phi)= \int_\Omega \phi\ dW = n\int_\Omega \phi(s) F(s) d\alpha(s)\;.$$
(here and below we work with spaces of real-valued functions).
The dual map $T^*: \cL(\Hilb) \to L_{\infty} (\Omega, \alpha)$ is given by $T^*(A)(s) = n((F(s),A))$.
Since $L_\infty \subset L_1$, we have an operator
$$\cE= \frac{1}{n}TT^*: \cL(\Hilb) \to \cL(\Hilb)\;,$$
\begin{equation}\label{eq-e} \cE(A) = n\int_{\Omega} ((F(s),A))F(s) d\alpha(s)\;.\end{equation}
Observe that $\cE$ is a unital trace-preserving completely positive map. In the terminology of \cite[Example 5.4]{Hayashi}, this is an example of an entanglement-breaking quantum channel.

Furthermore, set $$\cB=\frac{1}{n}T^*T: L_1(\Omega,\alpha) \to L_\infty(\Omega,\alpha)\;,$$ \begin{equation}\label{eq-b}\cB(\phi)(t)= n\int_\Omega \phi(s)((F(s),F(t))) d\alpha(s)\;.\end{equation}
Observe that the image of $\cB$ is finite-dimensional as $\cB$ factors through $\cL(\Hilb)$.

Write $(\phi,\psi):= \int_\Omega \phi\psi\, d\alpha$ for the scalar product on $L_2(\Omega,\alpha)$, and $\|\cdot\|$
for the associated norm.
Note that $\cB$ is defined as an operator on  $L_2(\Omega,\alpha)$ and its spectrum belongs to $[0,1]$, with $1$
being the maximal eigenvalue associated with the constant function.

Note now that positive eigenvalues of $\cE$ and $\cB$ coincide. Indeed, $T^*$ maps isomorphically
an eigenspace corresponding to a positive eigenvalue of $\cE$ to the eigenspace of $\cB$ corresponding
to the same eigenvalue. Write
$$1=\gamma_0 \geq \gamma_1 \geq \gamma_2 \geq \gamma_k
 \geq \dots \geq 0$$
for the eigenvalues of $\cB$ with multiplicites.

\begin{defin}\label{def-gap}
{\rm The non-negative number $$\gamma(W):= 1-\gamma_1\geq 0$$
is called {\it the spectral gap} of the POVM $W$.}
\end{defin}
With slight abuse of terminology, we sometimes refer to $\gamma(W)$ as the spectral gap of
operators $\cB$ and $\cE$.

Several aspects of this paper concern the positivity of such a
spectral gap, and are related to the theory of Markov chains with
state space $\Omega$. The notion of the spectral gap, while seemingly being unnoticed in the context of POVMs, naturally appears in the study of Markov chains, where it is responsible
for the rate of convergence to the stationary measure. In Section \ref{subsec-BT}, Markov chains will provide
a useful link between Berezin-Toeplitz quantization
and semiclassical random walks studied by Lebeau and Michel \cite{LM},
to get a semi-classical estimation of the spectral gap of the
Bereztin-Toeplitz POVM. In fact, the result of Lebeau and Michel
can be applied directly to get a lower bound on this spectral gap,
giving a weak version of Theorem \ref{thm-quant}.
Finally, as we have already mentioned, POVMs play a central role in the mathematical
theory of quantum measurements, and, interestingly enough, Markov chains arise in the context of repeated quantum measurements, see Section \ref{sec-noise}. Let us recall some basic notions from the theory of Markov chains \cite{Bak, Ru}.
A {\it Markov kernel}  on $\Omega$ is a map $x \mapsto \sigma_x$
sending a point $x \in \Omega$ to a probability measure $\sigma_x$ on $(\Omega,\cC)$ such that $x \mapsto \sigma_x(A)$
is a measurable function for every $A \in \cC$. With every Markov kernel $\sigma$ one associates a Markov chain, i.e., a sequence of $\Omega$-valued random variables $\zeta_k$, $k = 0,1,\dots$ defined on the same probability space, such that for every $n$ and every sequence $x_i \in \Omega$ the conditional probabilities satisfy $$\mathbb{P}(\zeta_n \;|\; \zeta_{n-1}=x_{n-1},\dots, \zeta_0 = x_0) = \sigma_{x_{n-1}}\;.$$ If $\zeta_0$ is distributed according to a probability measure $\nu_0$ on $\Omega$, then $\zeta_1$ is distributed according to $\nu_1$ given by the formula $$\nu_1(A) = \int_\Omega \sigma_x(A) d\nu_0(x)\; \forall A \in \cC.$$ If $\nu_0=\nu_1$, we say that $\nu_0$ is a stationary measure for the Markov chain.

The Markov kernel is called {\it reversible} with respect to a measure $\nu$ on $\Omega$ if
$$d\nu(x) d\sigma_x(y) =  d\sigma_y(x) d\nu(y)\;,$$ as measures on $\Omega \times \Omega$. In this case
$\nu$ is a stationary measure of the Markov chain. Given a $\nu$-reversible Markov kernel $\sigma$ with the state space $\Omega$, define {\it the Markov operator} $\cA$ on $L_1(\Omega,\nu)$ by
\begin{equation}\label{eq-cP}
\cA(\phi)(x) = \int_\Omega \phi(y)d\sigma_x (y)\;.
\end{equation}
Note that $\cA$ preserves positivity: $\cA(\phi)\geq 0$ for $\phi \geq 0$, $\cA(1)=1$, and its operator
norm is $\leq 1$. The reversibility readily yields that the Markov operator $\cA$ is self-adjoint on $L_2(\Omega, \nu)$. Denote by $1^\bot$ the orthogonal complement to the constant function $1$
on $\Omega$, i.e., the space of functions with zero mean. Then $\cA$ preserves $1^\bot$. By definition, the spectral gap $\gamma(\cA)$ is defined as
\begin{equation}\label{gap-Markov}
\gamma(\cA) = 1- \|\cA|_{1^\bot}\| = \inf_{\phi \neq 0} \frac{(\phi -\cA\phi,\phi)}{(\phi,\phi)-(\phi,1)^2}\;.
\end{equation}
With this language, the operator $\cB$ given by \eqref{eq-b} is a Markov operator with the Markov kernel
\begin{equation}\label{eq-kernel} t \mapsto n((F(s),F(t)))d\alpha(s)\;.\end{equation}
It is reversible with respect to the stationary measure $\alpha$.


\section{Spectral gap for quantization}\label{subsec-BT}

\subsection{Berezin transform vs. Laplace-Beltrami operator} \label{subsec-BLB}
Pure POVMs naturally appear in the context of
Berezin-Toeplitz quantization of
closed K\"{a}hler manifolds $(X,\om)$, which are
\emph{quantizable} in
the sense that the cohomology class $[\om]$ of the
K\"{a}hler symplectic form $\om\in\Om^2(X,\R)$ is integral.
Recall that this last condition is equivalent to the
existence of a holomorphic Hermitian line bundle $(L,h)$
over $X$ whose Chern connection has curvature $-2\pi i\omega$.

Let us briefly recall the construction of this quantization
(see \cite{BMS,S,LF} for prelimiaries). Let $X$ be a
quantizable closed K\"{a}hler manifold with $\dim_\C X=d$,
and let $(L,h)$ be a holomorphic Hermitian line bundle
as above. Write $L^p$ for the $p$-th tensor power of $L$,
and $h^p$ for the Hermitian metric on $L^p$ induced by $h$,
for any $p\in\N^*$
\footnote{Our convention is that the set of natural numbers $\N$ contains $0$.
We write $\N^*$ for strictly positive natural numbers.}.
Then the Hilbert space of quantum states in the space $\Hilb_p$
of global holomorphic sections of $L^p$, together with
the $L_2$-inner product induced by the Hermitian metric $h^p$
on $L^p$ and the \emph{Liouville measure} $dv_X$
associated to the canonical volume form $\omega^d/d!$.
We set $n_p=\dim_\C\cH_p$.
The quantity
$\hbar=1/p$ plays the role of the Planck constant, so that
the classical limit is given by $p \to +\infty$.
For all $p\in\N^*$ large enough, we define a pure
$\cL(\Hilb_p)$-valued POVM on $X$ through its density
\eqref{eq-POVM-density} by the formula
\begin{equation}\label{BTPOVM}
dW_{p} = n_p\,F_p\,d\alpha_p\,,
\end{equation}
where
the map $F_p: X\to \cS(\Hilb_p)$ sends a point $x \in X$
to the \emph{coherent state projector} with kernel the space
of sections vanishing at $x\in X$, and where
the measure $\alpha_p$ is given at any $x\in X$ by
\begin{equation}\label{eq-alphahbar}
d\alpha_p(x) = \frac{R_p(x)}{n_p}dv_X(x)\;,
\end{equation}
with density $R_p : X \to \R$ called the
\emph{Rawnsley function}.
From the viewpoint of complex geometry, the map $F_p$ is
given by the Kodaira map and
the Rawnsley
function is given by the value on the diagonal
of the Bergman kernel, i.e. the Schwarz kernel with respect to
$dv_X$ of the orthogonal projection
$\Pi_p:L^2(X,L^p)\rightarrow\cH_p$.
By the Kodaira embedding theorem,
for all $p\in\N^*$ large enough, the map $F_p$
is well defined and injective,
and we have $R_p(x)\neq 0$ for all $x\in X$, so that the
$\cL(\Hilb_p)$-valued measure $W_p$ defines
a pure POVM in the sense of \cref{prel}, called the
\emph{Berezin-Toeplitz POVM}.

In this context, the operator
$\cB_p:= \frac{1}{n_p} T^*_p T_p$ given by formula \eqref{eq-b}
is known as the {\it Berezin transform}.
Recall that for any $p\in\N^*$, the operator
$\cB_p$ has a finite-dimensional image, and all its
eigenvalues lie in the interval $[0,1]$. There is a finite number of positive eigenvalues with multiplicities,
while $0$ has infinite multiplicity. Write $$1=\gamma_{0,p}
\geq \gamma_{1,p} \geq \gamma_{2,p}\geq \dots \geq \gamma_{k,p} \geq \dots \geq 0$$
for the eigenvalues of $\cB_p$ with multiplicities.

Let $\Delta f= -\text{div} \nabla f$ be the (positive) Laplace-Beltrami operator associated with the Kähler metric, acting
on functions on $X$ with eigenvalues
\begin{equation}\label{evLB}
0=\lambda_0 < \lambda_1 \leq\lambda_2\leq \dots\leq\lambda_k\leq \dots\;.
\end{equation}

\begin{thm}\label{thm-quant}
For every integer $k\in\N$, we have the following
asymptotic estimate as $p\to +\infty$,
\begin{equation}\label{eq-LB}
1-\gamma_{k,p} =\frac{1}{4\pi p}\lambda_k + \bigo(p^{-2})\;.
\end{equation}
Furthermore, every sequence in $p\in\N^*$ of
$L_2(X,\alpha_p)$-normalized eigenfunctions of $\cB_p$
corresponding to the eigenvalue $\gamma_{k,p}$ contains a
subsequence converging to an eigenfunction of the Laplace-Beltrami
operator corresponding to $\lambda_k$ in the $\CC^\infty$-sense.
\end{thm}

The proof of Theorem \ref{thm-quant}
is given in Section \ref{specsec}.
Note that in the context of \cref{prel}, \cref{thm-quant} is
equivalent to the same statement via $T^*_p$
for the operator $\cE_p:\cL(\cH_p)\to\cL(\cH_p)$ defined from
$\cB_p$ by the formula \cref{eq-e}. Let us emphasize also that the remainder
$\bigo(p^{-2})$ in \eqref{eq-LB} is not uniform in $k$.

\medskip

The Berezin transform $\cB_p$ and its associated operator $\cE_p$
have prominent cousins, the $Q_{K,p}$-operator and the $Q_p$-operator,
respectively introduced by Donaldson \cite[\S\,4]{D}
in the framework of his program of finding numerical approximation
to distinguished K\"{a}hler metrics
on complex projective manifolds.
They are defined as
\begin{equation}\label{QKdef}
\begin{split}
Q_{K,p}&= \frac{\text{Vol}(X)}{n_p} \iota_p T_p : L_1(X)  \to L_{\infty}(X)\;,\\
Q_p &= \frac{\text{Vol}(X)}{n_p} T_p \iota_p: \cL(\Hilb_p) \to \cL(\Hilb_p)\;,
\end{split}
\end{equation}
where  for any $p\in\N^*$, the map
$\iota_p:\cL(\cH_p)\fl L_\infty(X)$ has been defined
in \cite[\S\,2.2.1]{D} for all $A\in\cL(\cH_p)$ and $x\in X$
by the formula
\begin{equation}\label{eq-iota}
\iota_p(A)(x)=\sum_{j=1}^{n_p} h^{p}(As_j(x),s_j(x))\,,
\end{equation}
where $\{s_j\}_{j=1}^{n_p}$ is an orthonormal basis of $\cH_p$.
By
definition of the coherent state projector
$F_p:X\rightarrow \cS(\Hilb_p)$ and in the language
of Section \ref{prel},
equation \eqref{eq-iota}
reads
\begin{equation}\label{iota=RT}
\iota_p(A)(x)=\frac{1}{n_p}R_p(x)T_p^*(A)(x)\,.
\end{equation}
On the other hand, by their definitions \eqref{QKdef},
the non-vanishing parts of the spectra
of $Q_{K,p}$ and $Q_p$ are finite and coincide together with
their multiplicities.
Write the eigenvalues of $Q_p$ as
\begin{equation}
\beta_{0,p}
\geq\beta_{1,p}\geq\beta_{2,p}\geq\dots\geq\beta_{k,p}\geq\dots\;,
\end{equation}
and set $$p' : = \left(\frac{n_p}{\Vol(X)}\right)^{1/d}\;.$$ For some K\"{a}hler metrics of constant
scalar curvature, Donaldson considered the $Q_p$-operator
as a finite-dimensional approximation of the heat operator and
predicted (see p.\,611 in \cite{D}) that as $p\fl+\infty$,
the spectrum of $Q_p$
approximate the spectrum of $e^{-\frac{\Delta}{4\pi p'} }$,
and an approximation to eigenfunctions of $e^{-\frac{\Delta}{4\pi p'} }$
can be extracted from the eigenvectors of $Q_p$.
The following result, which follows from Theorem \ref{thm-quant}
using the classical
asymptotics of the Rawnsley function as $p\fl+\infty$,
confirms
Donaldson's prediction for all K\"{a}hler metrics of constant
scalar curvature. A detailed proof
is given in Section \ref{specsec}.

\begin{thm}\label{cor-quant}
Assume that the K\"{a}hler metric of $X$ has constant scalar
curvature.
For every integer $k\in\N$, we have the following
asymptotic estimate as $p\to +\infty$,
\begin{equation}\label{eq-LQ}
1-\beta_{k,p} =\frac{1}{4\pi p'}\lambda_k + \bigo(p^{-2})\;.
\end{equation}
Furthermore, for every sequence in $\{A_p\}_{p\in\N^*}$ of
normalized eigenvectors of $Q_p$ in $\cL(\cH_p)$
corresponding to the eigenvalue $\beta_{k,p}$ for all $p\in\N^*$,
there is a
subsequence of
\begin{equation}\label{ipAp}
\Big{\{} \frac{\iota_pA_p}{\|\iota_pA_p\|_p}\Big{\}}_{p\in\N^*}
\end{equation}
converging to an
eigenfunction of the Laplace-Beltrami
operator corresponding to $\lambda_k$ in the $\CC^\infty$-sense,
where $\|\cdot\|_p$ is the norm on $L_2(X,\alpha_p)$.
\end{thm}

\medskip We refer to \cite{Fin12,Keller16} for a related study of the asymptotic behaviour of the spectrum
of certain geometric operators arising in Donaldson's program.

\medskip Let us introduce the following useful notion  \cite{D,Fine-lectures}.

\medskip
\noindent
\begin{defin}\label{def-balanced} {\rm Let $(X,\omega)$ be a closed K\"{a}hler manifold, and let $(L,h)$ be a Hermitian
holomorphic line bundle over $X$
whose Chern connection has curvature $-2\pi i\omega$.
Fix a positive integer $p\in\N^*$ so that the Kodaira map
$X \to H^0(X,L^p)$ is an embedding.
We say that the data $(X,L^p,h^p)$ is {\it balanced} if the corresponding
Rawnsley function $R_p:X\rightarrow\R$ is constant. }
\end{defin}

\medskip

Note that for the balanced data $(X,L^p,h^p)$ the Berezin
transform $\cB_p$ and the $Q_{K,p}$-operator coincide, as well
as $\cE_p$ and the $Q_p$-operator. In that case, the result of Theorem
\ref{thm-quant} is relevant in \cite[\S\,4.3]{D}.
We refer the reader to \cite[\S\,4.1]{D} and to \cite[\S\,1.4.1]{Fin10}
for an  interpretation of these operators in terms of complex geometry
of $(X,L^p)$. Let us finally mention that the approximation of the heat
operator by the $Q_K$-operator has been explored
by Liu and Ma in \cite{LM07}, and that the analogue of the refined
Karabegov-Schlichenmaier expansion of Proposition \ref{KS}
for the $Q_K$-operator has been shown by Ma and Marinescu in
\cite[Th.\,6.1]{MM12}. Some ingredients of their approach
are instrumental for us.
\medskip

It follows from Theorem \ref{thm-quant} that the spectral gap of the Berezin-Toeplitz POVM
equals
\begin{equation}\label{eq-gap-repeated}
\gamma(W_p)= \frac{\hbar}{4\pi}\lambda_1  + \bigo(\hbar^2)\;, \;\; \hbar = 1/p\;.
\end{equation}
In particular, this yields that the eigenvalue $1$ of $\cB_p$ is
simple (i.e., has multiplicity $1$) for all
sufficiently large $p$.

\medskip\noindent\begin{exam}\label{exam-CP1} {\rm Take the projective line  $X= \C P^1=S^2$ of area $1$.
Let $L= \bigo(1)$ be the holomorphic line bundle over $X$ dual to the tautological one.
The quantum Hilbert space $\Hilb_p$ of global holomorphic sections of $L^p$ can be identified
with the $(p+1)$-dimensional space of homogeneous polynomials of degree $p$ of 2 variables.
A representation-theoretical argument (see \cite {gap-representations, D} and Remark \ref{rem-rps} below) shows that the
eigenvalue $\gamma_1$ of the Berezin transform equals $p/(p+2)$. The K\"{a}hler metric on $X$ has constant curvature. For such metrics the first eigenvalue $\lambda_1$  of the Laplace-Beltrami operator equals $8\pi/\text{Area} = 8\pi$. We get that
$$\gamma = 1-\gamma_1 = \frac{2}{p+2} = \frac{1}{4\pi p}\lambda_1 + \bigo(p^{-2})\;,$$
as predicted by \eqref{eq-gap-repeated}.
}
\end{exam}

The upper bound in \eqref{eq-gap-repeated} immediately follows from the Karabegov-Schlichenmaier asymptotic expansion \eqref{eq-asymp} of the Berezin transform \cite{KS01}
$$\cB_p(f)= f- \frac{1}{4\pi p}\Delta f + \bigo(p^{-2})\;,$$
for every smooth function $f$ on $X$, where the remainder $\bigo(p^{-2})$ depends on $f$.
Indeed, choosing $f$ to be
the $L_2(X,\alpha_p)$-normalized first eigenfunction of $\Delta$,
we see that
$$\gamma(W_p) \leq ((\id-\cB_p)f,f)_p \leq \frac{1}{4\pi p}\lambda_1 + \bigo(p^{-2}),$$
where $(\cdot,\cdot)_p$ is the scalar product on $L_2(X,\alpha_p)$.

The prototypical example illustrating a link between the Berezin transform and the Laplace-Beltrami operator
is the flat space $\R^{2n}$, where the Berezin transform $\cB_p$ simply coincides with the heat operator $e^{-\hbar \Delta/4\pi}$ (see \cite{B}). It would be interesting to explore the following problem motivated by
a conversation with J.-M. Bismut. Denote by $\chi(t)$ the indicator function of the interval $[0,1]$.

\medskip
\noindent
\begin{problem}\label{prob} Call a non-decreasing sequence $r(p)$
in $p \in \N^*$ \textup{admissible} if
$$\left\|(\cB_p-e^{-\frac{\Delta}{4\pi p}}) \chi\left(\frac{\Delta}{r(p)}\right)\right\| = \bigo(p^{-2})\;,$$
where the norm stands for the operator norm in $L_2$.
According to Theorem \ref{thm-quant}, the constant sequence $r(p) = C$ is admissible for all $C$. Is the sequence $r(p) = p^\tau$ admissible
for $\tau >0$? What is the maximal
possible growth rate of an admissible sequence?
 \end{problem}

Let us finally make a couple of comments
on the physical intuition behind the Berezin transform. It has been noted in the introduction that the Berezin transform can be defined as the composition of the quantization and the dequantization.
It is instructive to interpret it in terms of the quantization only.
Let $\sigma$ be a classical state, i.e. a Borel probability measure on $X$, and following \cite{CP1}, define its
quantization as
$$\Theta_p(\sigma) = \int_X F(x)d\sigma(x)\;,$$
where as earlier $F(x)$ stands for the coherent state projector at $x\in X$. Let further $f \in L_2(X)$ be a classical observable.
It was noticed in \cite[(11)]{CP1}
that the expectation $((T_p(f),\Theta_p(\sigma)))$ of the value of the quantized observable $T_p(f)$
in the quantized state $\Theta_p(\sigma)$ equals the classical expectation $\int_X \cB(f)\,d\sigma$
of the Berezin transform $\cB(f)$ in the classical state $\sigma$.
Thus in the context of Berezin-Toeplitz quantization,
we get another interpretation of the blurring
of quantization measured by $\cB$.
Furthermore, in view of Theorem \ref{thm-quant}, we know that
$\cB$ is a Markov operator
with strictly positive spectral gap. Thus it has unique
stationary measure $\alpha_p$ whose density against the phase volume is given by $R_p/n_p$, as in formula
\eqref{eq-alphahbar}. Interestingly enough, this provides an interpretation of the Rawnsley function
without appealing to a specific choice of coherent states.

\subsection{Comments on the proof}\label{subsec-comments}
The proof of Theorem \ref{thm-quant} occupies the rest of this section, and we will deduce Theorem \ref{cor-quant} as a
consequence of it in Section \ref{specsec}.
Our argument has the same structure as the one in a paper by Lebeau and Michel \cite{LM}
on the Markov operator associated to the semiclassical random walk on
manifolds. The key intermediate results are as follows:
\begin{itemize}
\item[{(i)}] An apriori estimate stating that for any eigenfunction $f$ of $\cB_p$ whose eigenvalue
is sufficiently bounded away from $0$, {\bf any} Sobolev norm $\|f\|_{H^q}$ is bounded
by $C_q  \|f\|_{L_2}$. See Lemma \ref{hjrefprop} below which is a counterpart of  Lemma 5 in \cite{LM}.
\item[{(ii)}] The operators $\cA_p:= p(\id-\cB_p)$ and $\frac{\Delta}{4\pi}$ turn out to be  $\sim p^{-1}$-close as
operators from $L_2$ to $H^q$ for the Sobolev space $H^q$ with {\bf some} sufficiently large $q$, see formula \eqref{KSSob} below which is a counterpart of formula (3.28) in \cite{LM}, and which can be considered as
a refinement of the expansion \eqref{eq-asymp}  obtained in \cite{KS01}.
\end{itemize}
Combining (i) and (ii) we conclude that, roughly speaking, eigenfunctions of $\cA_p$ as in (i) are ``approximate" eigenfunctions of the Laplacian, which eventually implies that the spectra of $\cA_p$ and $\Delta$ are close to one another,
which yields the desired result (see the ending of our proof which is parallel to the one in \cite{LM}).

Proving (i) and (ii) forms the main bulk of the work. In contrast to \cite{LM}, our proof does not involve micro-local
analysis. The main ingredients we use is the expansion of the Bergman kernel due to Dai, Liu and Ma \cite{DLM06}
(see Theorem \ref{BTasy}) and a comparison between the Berezin transform and the heat operator motivated by the work of Liu and Ma on Donaldson's $Q_K$-operator \cite{LM07} (see Proposition \ref{boundexp} below).

Finally, an acknowledgment is in order. After a weaker version of Theorem \ref{thm-quant} was posted and formula \eqref{eq-gap-repeated} was stated as a question, Alix Deleporte kindly shared with us his ideas concerning the proof
of \eqref{eq-gap-repeated}. He sent us notes \cite{Deleporte} containing a number of preliminary steps in the direction of (i) and (ii) above. While the original arguments of Deleporte dealt with the case of real-analytic K\"{a}hler manifolds and line bundles and were based on the asymptotic expansion from \cite{Rouby, Deleporte}, he informed us that they also could be adjusted to the $\CC^\infty$-case.

\subsection{Preparations}\label{subsecprep}

Recall that the measure $dv_X$ associated to the canonical
volume form $\om^d/d!$ is also the Riemannian volume form of $X$.
Let $\<\cdot,\cdot\>_{L_2}$ be the usual $L_2$-scalar product
on $\cinf(X,\C)$, and let
$\|\cdot\|_{L_2}$ be the associated norm.
For all $j\in\N$, let $e_j\in\cinf(X,\C)$ be the normalized
eigenfunction associated with
the $j$-th eigenvalue of the Laplace-Beltrami
operator, so that $\|e_j\|_{L_2}=1$ and
$\Delta e_j=\lambda_j e_j$ as in \eqref{evLB} for all
$j\in\N$. Then for any $f\in\cinf(X,\C)$, we have
the following equality in $L_2$,
\begin{equation}
f=\sum_{j=0}^{+\infty}\<f,e_j\>_{L_2}e_j.
\end{equation}
For any $F:\R\fl\R$ bounded, we define the
bounded operator $F(\Delta)$
acting on $L_2(X,\C)$ by the formula
\begin{equation}\label{calculfct}
F(\Delta)f=\sum_{i=0}^{+\infty}
F(\lambda_j)\<f,e_j\>_{L_2}e_j\;.
\end{equation}
The bounded operator $e^{-t\Delta}$ thus defined for all $t>0$
is called the \emph{heat operator}. For any $m\in\N$, let
$|\cdot|_{\CC^m}$ be a $\CC^m$ norm on $\cinf(X,\C)$.
The following result is classical
and can be found for example in \cite{Kannai}, \cite[Th.\,2.29, (2.8)]{BGV04}.

\begin{prop}\label{heatexpprop}
For any $m\in\N$, there exists $C_m>0$ such that for any
$f\in\cinf(X,\C)$ and all $t>0$, we have
\begin{equation}\label{heatexpfla}
|e^{-t\Delta}f-f+t\Delta f|_{\CC^m}\leq C_m t^2
|f|_{\CC^{m+4}}\;.
\end{equation}
\end{prop}

For any $m\in\N^*$, let $\|\cdot\|_{H^m}$ be a Sobolev norm
of order $m$ on $\cinf(X,\C)$. Using the elliptic estimates for the
Laplace-Beltrami operator, for $m$ even we define
$\|\cdot\|_{H^m}$ by
\begin{equation}\label{ellest}
\|f\|_{H^m} := \|\Delta^{m/2}f\|_{L_2} + \|f\|_{L_2}\;.
\end{equation}
Note that the Laplacian $\Delta$ is symmetric with respect to the corresponding scalar product on $H^m$. By convention, we set
$\|f\|_{H^0}:=\|f\|_{L_2}$.

Next, turn to the Berezin transform. Recall that the Hermitian
product on $L$ and the Riemannian measure $dv_X$ induce an
$L_2$-scalar product on sections of $L^p$ for any $p\in\N^*$,
and write $L_2(X,L^p)$ for the associated Hilbert space.
The central tool for the
study of the Berezin transform
is the Schwartz kernel $\Pi_p(x,y)$ of
the orthogonal projector $\Pi_p:L_2(X,L^p) \to \Hilb_{p}$,
called the
\emph{Bergman kernel}. Recall that for fixed $x$ and $y$, this is
an element of $L^p_x \otimes \bar{L}^p_y$, where $L^p_x$ denotes the fiber of $L^p$ at $x \in X$ and the bar stands for the conjugate line bundle. Since the bundle $L$ comes with a Hermitian metric, we can measure the point-wise norm
$|\Pi_p(x,y)|$.
By Corollary 9.1.4\,(2) in \cite{LF}, we have that $|\Pi_p(x,y)| = |\langle \xi_{x,p},\xi_{y,p} \rangle|$, where $\xi_{x,p}$ is the non-normalized
\emph{coherent state} at $x\in X$ defined up to a phase factor
(see e.g. \cite{CP,LF} for the definition).
The Rawnsley function $R_p$
is given by $R_p(x) = |\xi_{x,p}|^2$, and thus satisfies
$R_p(x)=\Pi_p(x,x)$.
Since $F_p(x)$ is the projector to $\xi_{x,p}$, we have that
$$|\Pi_p(x,y)|^2 = ((F_p(x),F_p(y))) R_p(x)R_p(y)\;.$$ It follows from \eqref{eq-kernel} and
\eqref{eq-alphahbar} that
\begin{equation}\label{Bpfla}
(\cB_p f)(x) = n_p \int_X ((F_p(x),F_p(y)))f(y) d\alpha_p(y)  = \frac{1}{R_p(x)} \int_X |\Pi_p(x,y)|^2 f(y) dv_X(y)\;,
\end{equation}
so that the Schwarz kernel of $\cB_p$ with respect to $dv_X$
is given by
\begin{equation}
\label{eq-schwarzkernelberezin}
\cB_p(x,y)= \frac{|\Pi_p(x,y)|^2}{R_p(x)}\;.
\end{equation}
Let $\|\cdot\|_p$ be the norm on $L_2(X,\alpha_p)$.
From the classical asymptotic
expansion of $R_p$ as $p \to +\infty$, we get a constant $C>0$
such that
\begin{equation}\label{diagexp}
\left(\frac{1}{\Vol(X)}- Cp^{-1}\right)\|\cdot\|_{L_2}
\leq\|\cdot\|_p
\leq\left(\frac{1}{\Vol(X)}+ Cp^{-1}\right)\|\cdot\|_{L_2}\;.
\end{equation}

\subsection{Asymptotic expansion of the Berezin transform}
For a comprehensive account on the off-diagonal expansion of the
Bergman kernel as well as tools of
Berezin-Toeplitz quantization
in this context, we refer to \cite{MM07}.

We always assume that $p\in\N^*$ is as large as needed.
For any $s>0$, we use the notation $O(p^{-s})$ as
$p\fl+\infty$ in the usual sense, uniformly in $\CC^m$-norm
for all $m\in\N^*$.
The notation $O(p^{-\infty})$ means $O(p^{-s})$
for any $s>0$.

Let $\epsilon_0>0$ be smaller than the injectivity radius of $X$.
Fix a point $x_0 \in X$, and let
$Z=(Z_1,...,Z_{2d})\in\R^{2d}$ with $|Z|<\epsilon_0$ be geodesic
normal coordinates around $x_0$, where $|\cdot|$ is the Euclidean
norm of $\R^{2d}$.
In these coordinates, the canonical volume form is given by
\begin{equation}\label{defkappa}
dv_X(Z)=\kappa_{x_0}(Z)dZ\;,
\end{equation}
with $\kappa_{x_0}(0)=1$.
For any kernel $K(\cdot,\cdot)\in\cinf(X\times X,\C)$, we write
$K_{x_0}(\cdot,\cdot)$ for its image in these coordinates, and
we write $|K_x|_{\CC^m(X)}$ for the $\CC^m$-norm of the family of functions $K_x$ with respect to $x \in X$.

Let $d^X$ be the Riemannian
distance on $X$.
We will derive  Theorem \ref{thm-quant} as a consequence of the
following asymptotic  expansion as $p \to +\infty$ of the
Schwartz kernel of the Berezin transform.

\begin{thm}\label{BTasy}
For any $m\,,k\in\N$ and $\epsilon>0$, there is
$C>0$ such that for all $p\in\N^*$ and
$x,y\in X$ satisfying $d^X(x,y)>\epsilon$,
\begin{equation}\label{thetafla}
|\cB_p(x,y)|_{\CC^m}\leq Cp^{-k}\;.
\end{equation}
For any $m, k\in\N$, there is
$N\in\N$ and $C>0$ such that for any
$x_0\in X,\,|Z|,|Z'|<\epsilon_0$
and for all $p\in\N^*$, we have
\begin{multline}\label{BTexp}
\Big|p^{-d} B_{p,x_0}(Z,Z')
-\sum_{r=0}^{k-1} p^{-r/2}J_{r,x_0}(\sqrt{p}Z,\sqrt{p}Z')
\exp(-\pi p|Z-Z'|^2)\kappa_{x_0}^{-1}(Z')\Big|_{\CC^m(X)}\\
\leq Cp^{-\frac{k}{2}}(1+\sqrt{p}|Z|+\sqrt{p}|Z'|)^N
\exp(-\sqrt{p}|Z-Z'|/C)+O(p^{-\infty})\;,
\end{multline}
where $\{J_{r,x_0}(Z,Z')\}_{r\in\N}$ is a family of polynomials
in $Z,Z'\in\R^{2n}$ of the same parity as $r$,
depending smoothly on $x_0\in X$. Furthermore,
for any $Z,Z'\in\R^{2n}$ we have
\begin{equation}\label{|J|0}
J_{0,x_0}(Z,Z')=1\quad\text{and}\quad J_{1,x_0}(Z,Z')=0\;.
\end{equation}
\end{thm}

\medskip
This readily follows from formula \eqref{eq-schwarzkernelberezin}
expressing the Schwarz kernel of the Berezin transform via the Bergman kernel $\Pi_p$ and
the analogous result of Dai, Liu and Ma in
\cite[Th.\,4.18']{DLM06} for the Bergman kernel.

For any $x\in X$, let $B^X(x,\epsilon_0)$ be the geodesic ball
of radius $\epsilon_0>0$ around $x$, and write $B(0,\epsilon_0)
\subset\R^{2d}$
for the Euclidean ball of radius $\epsilon_0$ around $0$.
The following proposition is a refinement of the
Karabegov-Schlichenmaier expansion \cite[(1.2)]{KS01}
of the Berezin transform, where we make explicit the remainder
term.

\begin{prop}\label{KS}
For any $m\in\N$, there exists $C_m>0$ such that for any
$f\in\cinf(X,\C)$ and all $p\in\N^*$, we have
\begin{equation}\label{KSexp}
\left|\cB_pf-f+\frac{\Delta}{4\pi p} f\right|_{\CC^m}\leq
\frac{C_m}{p^2}|f|_{\CC^{m+4}}\;.
\end{equation}
\end{prop}
\begin{proof}
For any $x\in X$, write $f_x$ for the
image of $f$ restricted to $B^X(x,\epsilon_0)$ in
normal coordinates around $x$.
From \eqref{thetafla}, we know that
for any $\epsilon>0$ and $x\in X$,
\begin{equation}\label{KStheta}
\begin{split}
(\cB_pf)(x)&=\int_X \cB_p(x,y)f(y)dv_X(y)\\
&=\int_{B^X(x,\epsilon_0)}\cB_p(x,y)f(y)dv_X(y)
+O(p^{-\infty})\,|f|_{\CC^0}\\
&=\int_{B(0,\epsilon_0)}B_{p,x}(0,Z)f_x(Z)\kappa_x(Z)dZ
+O(p^{-\infty})\,|f|_{\CC^0}\;.
\end{split}
\end{equation}
For any $k\in\N^*$ and $m\in\N$, we will use the following
Taylor expansion of $f_x$ up to order $k-1$,
for all $p\in\N^*$ and $|Z|<\epsilon_0$,
\begin{equation}\label{Taylorfk}
\begin{split}
f_{x}&(Z)  =\sum_{0\leq|\alpha|\leq k-1}\frac{\partial^{|\alpha|} f_{x}}{\partial Z^\alpha}\frac{Z^\alpha}{\alpha!}
+O_m(|Z|^{k})|f|_{\CC^{m+k}}\\
& =\sum_{0\leq |\alpha|\leq k-1}p^{-\frac{|\alpha|}{2}}\frac{\partial^{|\alpha|} f_{x}}{\partial Z^\alpha}\frac{(\sqrt{p}Z)^\alpha}{\alpha!}+p^{-\frac{k}{2}}O_m(|\sqrt{p}Z|^{k})|f|_{\CC^{m+k}}\;,
\end{split}
\end{equation}
where $O_m$ means that the expansion is uniform in $x\in X$
as well as all its derivatives up to order $m\in\N$, and does not
depend on $f$.

We will compute the asymptotic expansion as $p\fl+\infty$
of \eqref{KStheta} using the Taylor expansion \eqref{Taylorfk}
of $f$ and the asymptotic expansion \eqref{BTexp}
of the Berezin transform up to order $3$.
First, using the fact that $\cB_p 1=1$ for all $p\in\N^*$,
we know that the polynomials $J_{r,x}(Z,Z')$ of the asymptotic
expansion \eqref{BTexp} of the Berezin transform satisfy
\begin{equation}\label{KSest3}
\int_{\R^{2n}}J_{r,x}(0,Z)\exp(-\pi p|Z|^2)dZ=0\;,
\end{equation}
for all $x\in X$ and $r\in\N^*$.
On another hand, recall from \eqref{|J|0} that $J_{0,x}\equiv 1$ and
$J_{1,x}\equiv 0$ for all $x\in X$.
Using the parity of
Gaussian functions, a change of variable $Z\mapsto Z/\sqrt{p}$
and the Taylor expansion \eqref{Taylorfk} for $k=4$, we get that
\begin{multline}\label{KSest1}
p^d\int_{B(0,\epsilon_0)}\exp(-\pi p|Z|^2)f_x(Z)dZ\\
=f(x)
+p^{-1}\sum_{j=1}^{2n}\frac{\partial^2 f_{x}}{\partial Z^2_j}(0)
\int_{\R^{2n}}\frac{Z_j^2}{2}\exp(-\pi |Z|^2)dZ+
|f|_{\CC^{m+4}}O_m(p^{-2})\\
=f(x)-p^{-1}\frac{\Delta}{4\pi}f(x)+|f|_{\CC^{m+4}}O_m(p^{-2})\;.
\end{multline}
Recall that $J_{r,x}(0,Z)\in\C[Z]$ is a polynomial
in $Z\in\R^{2n}$ of the same parity than $r\in\N$, so that using
\eqref{Taylorfk}, \eqref{KSest3} and the parity of
Gaussian functions, we get in the same way
\begin{equation}\label{KSest2}
\begin{split}
p^d\int_{B(0,\epsilon_0)}J_{2,x}(0,\sqrt{p}Z)
&\exp(-\pi p|Z|^2)f_x(Z)dZ\\
=f(x)\int_{\R^{2n}}&J_{2,x}(0,Z)\exp(-\pi|Z|^2)dZ+O_m(p^{-1})
|f|_{\CC^{m+2}}\\
&=O_m(p^{-1})|f|_{\CC^{m+2}}\;,\\
p^d\int_{B(0,\epsilon_0)}J_{3,x}(0,\sqrt{p}Z)
&\exp(-\pi p|Z|^2)f_x(Z)dZ=O_m(p^{-1/2})|f|_{\CC^{m+1}}\;.
\end{split}
\end{equation}
Finally, again using a change of variable $Z\mapsto Z/\sqrt{p}$,
we get for any $N\in\N^*$ and $p\in\N^*$,
\begin{equation}\label{KSest4}
p^d\int_{B(0,\epsilon_0)}(1+|\sqrt{p}Z|)^N\exp(-\sqrt{p}|Z|/C)f_x(Z)dZ
=O_m(1)|f|_{\CC^m}\;.
\end{equation}
This completes the proof of \eqref{KSexp}.
\end{proof}

In view of Proposition \ref{heatexpprop} and Proposition \ref{KS},
it is natural
to compare the Berezin transform with the heat operator
by setting $t=(4\pi p)^{-1}$.
This leads to the following result, which is essentially a
refinement of \cite[Th.\,0.1]{LM07}.

\begin{prop}\label{boundexp}
For any $m\in\N$, there exists $C_m>0$ such that for any
$f\in\cinf(X,\C)$ and all $p\in\N^*$, we have
\begin{equation}\label{boundexpfla}
\left\|(e^{-\frac{\Delta}{4\pi p}}-\cB_p)f\right\|_{H^m}\leq
\frac{C_m}{p}\|f\|_{H^m}\;.
\end{equation}
\end{prop}
\begin{proof}
Set $S_p:=e^{\frac{\Delta}{4\pi p}}-\cB_p$, which acts on
$L_2(X,\C)$ for all $p\in\N^*$ and admits a smooth
Schwartz kernel $S_p(\cdot,\cdot)$ with respect to $dv_X$.
Comparing \cref{BTasy} with the classical asymptotic expansion
of the heat kernel,
as given for example in \cite[Th.\,2.29]{BGV04},\cite{Kannai}, we see that
\begin{equation}\label{thetaflaSp}
S_p(x,y)=O(p^{-\infty})\;,
\end{equation}
for all $x,y\in X$ satisfying $d^X(x,y)>\epsilon_0$,
and using the formula \eqref{|J|0} for the
first two coefficients, we get for any $m\in\N$ a constant
$C>0$ and $N\in\N$ such that
\begin{multline}\label{B-eexp}
\left|S_{p,x_0}(Z,Z')\right|_{\CC^m(X)}\leq Cp^{-1}
(1+\sqrt{p}|Z|+\sqrt{p}|Z'|)^N
\exp(-\sqrt{p}|Z-Z'|/C)\\
+O(p^{-\infty})\;.
\end{multline}
Let us first show \eqref{boundexpfla} for $m=0$.
For any $f\in\cinf(X,\C)$ and any $\epsilon>0$, by
Cauchy-Schwarz inequality and \eqref{thetaflaSp} for $S_p$,
we get the following version of the Schur test for all $p\in\N^*$,
\begin{equation}\label{Schurtest}
\begin{split}
\|S_p f\|_{L_2}^2&\leq\int_X\left(\int_X |S_p(x,y)|\,dv_X(y)\right)
\left(\int_X |S_p(x,y)|\,|f(y)|^2\,dv_X(y)\right)dv_X(x)\\
&\leq\sup_{x\in X}\left(\int_X|S_p(x,y)|\,dv_X(y)\right)
\sup_{y\in X}\left(\int_X|S_p(x,y)|\,dv_X(x)\right)\|f\|_{L_2}^2\\
&\leq\sup_{x\in X}\left(\int_{B(x,\epsilon_0)}|S_p(x,y)|\,dv_X(y)\right)
\sup_{y\in X}\left(\int_{B(x,\epsilon_0)}|S_p(x,y)|\,dv_X(x)\right)\|f\|_{L_2}^2\\
&\quad\quad\quad\quad\quad\quad\quad\quad\quad\quad\quad\quad\quad\quad\quad\quad\quad\quad\quad\quad\quad\quad
+O(p^{-\infty})\|f\|_{L_2}^2\;.
\end{split}
\end{equation}
Then \eqref{boundexpfla} for $m=0$ follows from
\eqref{B-eexp} with $Z=0$ or $Z'=0$ respectively,
as in \eqref{KSest4}.

To deal with the case of arbitrary $m\in\N^*$, let us assume
by induction that \eqref{boundexpfla} is satisfied for $m-1$.
Considering the estimates \eqref{thetaflaSp} and \cref{B-eexp} with
corresponding $m\in\N^*$, note that for any differential
operator $D_x$ of order $m$ in
$x\in X$, there exists a differential operator $D'_{x,y}$ in
$x,y\in X$ of total order $m$ but of order at most $m-1$
in $x\in X$, such that the operator $S_p^{(m)}$
defined through its kernel for all $x,y\in X$ by
\begin{equation}
S_p^{(m)}(x,y):=D_xS_p(x,y)+D'_{x,y}S_p(x,y)
\end{equation}
also satisfies \eqref{thetaflaSp} and \eqref{B-eexp}.
Then for all $x\in X$ and
$p\in\N^*$, we get
\begin{equation}
\begin{split}
\int_X D_xS_p(x,y)f(y)dv_X(y)&=
-\int_X(D'_{x,y}S_p(x,y))f(y)dv_X(y)+(S_p^{(m)}f)(x)\\
&=\int_XD_x'S_p(x,y)(D_y''f(y))dv_X(y)+(S_p^{(m)}f)(x)\;,
\end{split}
\end{equation}
where $D'_x$ and $D''_y$ are differential operators, respectively
in $x$ and in $y$,
obtained from $D'_{x,y}$ using a partition of unity and integration
by parts in local charts, so that in particular $D'_x$
is of order $m-1$ in $x\in X$.
Then using the induction hypothesis, the inequality
\eqref{boundexpfla} for $m$
follows from the same inequality for $m-1$
replacing $f$ by any number of derivatives of $f$,
and from the estimates \eqref{B-eexp} and \eqref{Schurtest}
for $S_p^{(m)}$ in the same way than before.
\end{proof}

\subsection{Spectrum}\label{specsec}

Recall that $\|\cdot\|_p$ denotes the norm on $L_2(X,\alpha_p)$.
In this section, we consider a sequence $\{f_p\}_{p\in\N^*}$,
with $f_p\in\cinf(X,\C)$ such that
\begin{equation}\label{efBp}
\|f_p\|_p=1\;,\quad \cB_pf_p=\mu_p f_p\;,
\end{equation}
for some $\mu_p\in\Spec(\cB_p)$ for all $p\in\N^*$.
The following estimate is crucial
for the proof of Theorem \ref{thm-quant}.

\begin{lemma}\label{hjrefprop}
Assume that the sequence $\{p(1-\mu_p)\}_{p\in\N^*}$ is
bounded by some constant $L>0$. Then
for all $m\in\N$, there exists $C_{L,m}>0$
such that for all $p\in\N^*$, we have
\begin{equation}\label{hjrefined}
\|f_p\|_{H^{2m}}\leq C_{L,m}\;.
\end{equation}
\end{lemma}
\begin{proof}
Note that \eqref{hjrefined} is automatically verified for $m=0$
by \eqref{diagexp} and \eqref{efBp}.
By induction on $m\in\N$, let us assume that \eqref{hjrefined}
is satisfied for $m-1$. Let us write
\begin{equation}\label{deltFest}
\begin{split}
p(e^{-\frac{\Delta}{4\pi p}}-\cB_p)f_p&=p(1-\mu_p)f_p-
p(\id-e^{-\frac{\Delta}{4\pi p}})f_p\\
&=p(1-\mu_p)f_p-\Delta F(\Delta/p)f_p\;,
\end{split}
\end{equation}
where the bounded operator $F(\Delta/p)$ acting on $L_2(X,\C)$
is defined as in \eqref{calculfct} for the continuous function
$F:\R\fl\R$ given for any $s\in\R^*$ by
$F(s)=4\pi(1-e^{-s/4\pi})/s$.
As $|p(1-\mu_p)|<L$ for all $p\in\N^*$, by \cref{boundexp} and
formula \eqref{ellest} for $\|\cdot\|_{H^{2m}}$,
this gives a constant $C_m>0$ such that
\begin{equation}\label{sobfest}
\|F(\Delta/p)f_p\|_{H^{2m}}\leq C_m\|f_p\|_{H^{2m-2}}\;.
\end{equation}
On the other hand, note that by hypothesis, we have $\mu_p\fl 1$
as $p\fl+\infty$. Using \cref{boundexp} again, we then get
$\epsilon_m>0$ and $p_m\in\N^*$ such that for all $p>p_m$,
\begin{equation}\label{hjdeltFest}
\begin{split}
\|F(\Delta/p)f_p\|_{H^{2m}}&\geq
\|F(\Delta/p)f_p+(\cB_p-e^{-\frac{\Delta}{4\pi p}})f_p\|_{H^{2m}}
-\|(\cB_p-e^{-\frac{\Delta}{4\pi p}})f_p\|_{H^{2m}}\\
&\geq\inf_{s>0}\,\{F(s)+\mu_p-e^{-s/4\pi}\}\,\|f_p\|_{H^{2m}}
-C_m p^{-1}\|f_p\|_{H^{2m}}\\
&\geq\epsilon_m\|f_p\|_{H^{2m}}\;.
\end{split}
\end{equation}
This together with \eqref{sobfest} gives \eqref{hjrefined}.
\end{proof}

\medskip\noindent{\bf Proof of Theorem \ref{thm-quant}.}
For any $f\in\cinf(X,\C)$, by Proposition \ref{KS},
we get that
\begin{equation}\label{KSSob}
\left\|p(\id-\cB_p)f-\frac{\Delta}{4\pi}f\right\|_{L_2}
\leq Cp^{-1}|f|_{\CC^{4}}
\leq Cp^{-1}\|f\|_{H^{q}}\;,
\end{equation}
with $q$ even and large enough. The inequality on the right follows from Sobolev embedding theorem, and
the same is true in $L_2(X,\alpha_p)$-norm by \eqref{diagexp}.
Let now $j\in\N$ be fixed.
If $e_j\in\cinf(X,\C)$ satisfies
$\Delta e_j=\lambda_j e_j$ and $\|e_j\|_{L_2}=1$,
then by  \eqref{KSSob} we get
$C_j>0$ not depending on $p\in\N^*$ such that
\begin{equation}\label{estevdelt}
\left\|p(\id-\cB_p)e_j-\frac{\lambda_j}{4\pi}e_j\right\|
_p\leq C_j p^{-1}\;.
\end{equation}
Thus if $m_j\in\N$ is the multiplicity of $\lambda_j$ as an
eigenvalue of $\Delta$, the estimate \eqref{estevdelt} for all
eigenfunctions of $\Delta$ associated with $\lambda_j$
gives a constant $C>0$ such that
\begin{equation}\label{spec>}
\#\left(\Spec\big(p(\id-\cB_p)\big)\cap\left[\frac{\lambda_j}{4\pi}-Cp^{-1},\frac{\lambda_j}{4\pi}+Cp^{-1}\right]\right)
\geq m_j\;.
\end{equation}
This immediately follows from the variational principle for the operator $p(\id-\cB_p)-\frac{\lambda_j}{4\pi}\id$
acting on $L_2(X,\alpha_p)$.

Consider now for every $p\in\N^*$ a normalized eigenfunction
$f_p\in\cinf(X,\C)$ of $\cB_p$ as in \eqref{efBp}
such that the associated sequence
$\{p(1-\mu_p)\}_{p\in\N^*}$ of eigenvalues of $p(\id-\cB_p)$
is bounded. Combining Lemma \ref{hjrefprop} with the right
inequality in \eqref{KSSob}, we get $C>0$ such that
\begin{equation}\label{specinv>}
\|p(1-\mu_p)f_p-\Delta f_p\|_{L_2}\leq Cp^{-1}\;.
\end{equation}
In particular, we get that
\begin{equation}\label{eq-spec-est}
\textup{dist}\left(p(1-\mu_p),\Spec\Delta\right)\leq Cp^{-1}\;.
\end{equation}

Finally, let us show that there exists $p_0\in\N^*$ such that
\eqref{spec>} is in fact an equality
for $p>p_0$. To this end, let $l\in\N^*$ with $l\geq m_j$
be such that for all $p\in\N^*$, there exists
an orthonormal family $f_{k,p},\,1\leq k\leq l,$
of eigenfunctions of $\cB_p$ in $L_2(X,d\alpha_p)$ with associated
eigenvalues $\mu_{k,p}\in\R,\,1\leq k\leq l,$
satisfying
\begin{equation}
p(1-\mu_{k,p})\in[\lambda_j-Cp^{-1},\lambda_j+Cp^{-1}]\;,~~
\text{for all}~~1\leq k\leq l\;.
\end{equation}
As the inclusion of the Sobolev space $H^q$ in $H^{q-1}$ is compact,
by Lemma \ref{hjrefprop} there exists a subsequence of
$\{f_{k,p}\}_{p\in\N^*}$ converging
to a function $f_k$ in $H^{q-1}$-norm,
for all $1\leq k\leq l$.
In particular, taking $q>2$,
the family $f_k,\,1\leq k\leq l$,
is orthonormal in $L_2(X,\C)$
and satisfies $\Delta f_k=\lambda_j f_k$ for all $1\leq k\leq l$
by
\eqref{specinv>}. By definition of the multiplicity $m_j\in\N$
of $\lambda_j$, this forces $l=m_j$.

Let us sum up our findings. First, the equality
$$
\#\left(\Spec\big(p(\id-\cB_p)\big)\cap\left[\frac{\lambda_j}{4\pi}-Cp^{-1},\frac{\lambda_j}{4\pi}+Cp^{-1}\right]\right)
= m_j\;,
$$
where $m_j$ is the multiplicity of $\lambda_j$ as the eigenvalue of $\Delta$, together with \eqref{eq-spec-est} readily yields the first statement of the theorem:
$$1-\gamma_{k,p} =\frac{1}{4\pi p}\lambda_k + \bigo(p^{-2})\;.$$
Second, observe that we got a subsequence of $f_{k,p}$, $p \in \N^*$ converging to $f_k$ in the Sobolev $H^{q-1}$ sense, where $q$ even
can be chosen arbitrarily large. By the Sobolev embedding theorem, this yields a subsequence which
$\CC^l$-converges to $f_k$ with arbitary $l$. Iterating this argument for this subsequence we get that
there exists a sequence $p_l \to +\infty$ such that
$$|f_{k,p_l} - f_k|_{\CC^l} \leq 1/l\;,$$
which means that $f_{k,p_l}$ converges to $f_k$ in the $\CC^\infty$-sense.
This completes the proof.\qed\\

\medskip\noindent{\bf Proof of Theorem \ref{cor-quant}.}
For any $p\in\N^*$, using equation \eqref{iota=RT} for
$\iota_p$ and combining the definition \eqref{QKdef} of
the $Q_{K,p}$-operator with
formula \eqref{Bpfla} for the Berezin
transform $\cB_p= \frac{1}{n_p} T^*_p T_p$
acting on $f\in\cinf(X,\C)$, we get
\begin{equation}\label{Q_K}
(Q_{K,p} f)(x)= \frac{\text{Vol}(X)}{n_p} R_p(x) \cB_p(f)(x)
   = \frac{\text{Vol}(X)}{n_p}
\int_X |\Pi_p(x,y)|^2\,f(y)\,dv_X(y)\;.
\end{equation}
We will show that when the scalar curvature is constant,
the analogue of Theorem \ref{thm-quant} holds for this operator.
As $p'/p = 1 + \bigo(p^{-1})$ by the Riemann-Roch theorem,
this will imply Theorem \ref{cor-quant} via the morphism $\iota_p$
which relates $Q_p$ with $Q_{K,p}$, see \eqref{QKdef}.

Recall that $R_p:X\to\R$ denotes the Rawnsley function, and
that $n_p=\dim_\C\HH_p$.
By the classical asymptotic expansion of the Bergman kernel,
which can be found for example in \cite[\S\,4.1.1]{MM07}, we know
that when the scalar curvature is constant, we have
\begin{equation}
\frac{\Vol(X)}{n_p}R_p=1+O(p^{-2})\;.
\end{equation}
As this expansion holds in $\cC^m$-norm for all $m\in\N^*$
and by the definition $\cB_p$ and $Q_{K,p}$ in formulas
\eqref{Bpfla} and \eqref{Q_K} respectively, we get a constant
$C_m>0$ for any $m\in\N^*$ such that
\begin{equation}\label{QK-Bp}
\|Q_{K,p}-\cB_p\|_{H^m}\leq C_m p^{-2}.
\end{equation}
It is then easy to see that Lemma \ref{hjrefprop} holds
for any sequence $\{f_p\}_{p\in\N^*}$
with $f_p\in\cinf(X,\C)$ such that
\begin{equation}\label{efQKp}
\|f_p\|_{L_2}=1\;,\quad Q_{K,p}f_p=\mu_p f_p\;,
\end{equation}
with $\{p(1-\mu_p)\}_{p\in\N^*}$ bounded, simply using the estimate
\eqref{QK-Bp} to replace
$B_p$ by $Q_{K,p}$ in \eqref{deltFest} and \eqref{hjdeltFest}.
We can then follow the proof
of Theorem \ref{thm-quant} above to get the same result for
$Q_{K,p}$,  using the estimate \eqref{QK-Bp} to replace
$\cB_p$ by $Q_{K,p}$ in \eqref{KSSob} and \eqref{estevdelt}, and
using \eqref{diagexp} to replace
$\|\cdot\|_p$ by $\|\cdot\|_{L_2}$ in \eqref{estevdelt}.
Finally, the form \eqref{ipAp} for the normalized sequence
of eigenfunction of $Q_{K,p}$ follows from the fact that
$\iota_pQ_p = Q_{K,p}\iota_p$ by definition \eqref{QKdef}
of $Q_p$ and $Q_{K,p}$.
This completes the proof.
\qed\\

\begin{rem}
\rm{
\cref{thm-quant} can be extended to the case of a general
closed symplectic manifold $(X,\omega)$ of real dimension $2d$,
and $(L,h,\nabla)$ a Hermitian line
bundle with Hermitian connection $\nabla$ of curvature
$-2\pi i\om$.
In fact, one can in general consider the following
\emph{renormalized Bochner Laplacian} acting on $\cinf(X,L^p)$ for
any $p\in\N^*$, first introduced by Guillemin
and Uribe \cite{GU88},
\begin{equation}
\Delta_p:=\Delta^{L^p}-2\pi dp\;,
\end{equation}
where $\Delta^{L^p}$ stands for the usual Bochner Laplacian
on $L^p$. By \cite[Th.\,2.a]{GU88}, the spectrum of
$\Delta_p$ is contained in $I\,\cup\,(C_1p-C_2,+\infty)$
for all $p\in\N^*$, for some $C_1,\,C_2>0$ and some interval
$I\subset\R$ containing $0$.
We can then consider $\Pi_p$ as the associated spectral projection
corresponding to $I$ and set $\cH_p=\Im(\Pi_p)$.
Using the work \cite{MM08} of Ma and Marinescu on the kernel of
$\Pi_p$, we can then consider the Berezin-Toeplitz POVM
of Section \ref{subsec-BLB}. By \cite[(2.31),\,(3.2)]{LMM16},
the Berezin transform admits an asymptotic expansion
similar to \cref{BTasy}, except for the formula
\eqref{|J|0}, where we only have $J_{1,x_0}(0,Z')=0$ for all
$Z'\in\R^{2d}$ as a consequence of \cite[Lem.6.1,\,Lem.6.2]{ILMM17}
(see also \cite[(2.32)]{MM08}).
Then \cref{KS,boundexp}
hold, and it is
straightforward to adapt the rest of the proof of \cref{thm-quant}
in \cref{specsec}.
Note that the corresponding estimates in \cref{KS} and
\cref{boundexp} can be seen as refinements of \cite{LMM16}.
}
\end{rem}

\begin{rem}\label{weightedBT}
\rm{On the other hand, \cref{thm-quant} can be extended
to the case of \emph{weighted Berezin transforms}, introduced
by Englis in \cite{E00} in the case of pseudoconvex domains.
This corresponds to the case where one replaces the canonical
volume form $\om^d/d!$ by a general smooth volume form $\nu$
in the setting of
Section \ref{subsec-BLB}. In fact, let us consider the
Hilbert space $\cH_{\nu,p}$ of global holomorphic sections
of $L^p$ together with the $L_2$-inner product with respect
to the measure $d\nu$ instead of the Liouville measure $dv_X$.
Then using the trick of Ma and
Marinescu in \cite[\S\,4.1.9]{MM07}, we can define for any
$p\in\N^*$ large enough the $\cL(\cH_{\nu,p})$-valued POVM
\begin{equation}\label{Wnup}
dW_{\nu,p}=n_pF_{\nu,p}d\alpha_{\nu,p}\,,
\end{equation}
where $F_{\nu,p}:X\rightarrow\cS(\cH_{\nu,p})$ is the
map sending $x\in X$ to the
orthogonal projector with kernel the space of sections
vanishing at $x\in X$ and $\alpha_{p,\nu}$
is given by
\begin{equation}\label{eq-alphahbarnu}
d\alpha_{\nu,p}(x) = \frac{R_{\nu,p}(x)}{n_p}d\nu(x)\;,
\end{equation}
where $R_{\nu,p}:X\rightarrow\R$ is the weighted
Rawnsley function, given by the value on the diagonal
of the Schwarz kernel with respect to $\nu$
of the orthogonal projector operator
$\Pi_{\nu,p}:L^2(X,L^p,d\nu)\rightarrow\cH_{\nu,p}$.
Using \cite[\S\,4.19]{MM07} again as well
as the general version of the expansion \cref{BTasy} given
in \cite[Th.\,4.18']{DLM06}, the proof of \cref{thm-quant}
above extends verbatim to this case, to get the estimate
\begin{equation}\label{eq-LBnu}
1-\gamma_{\nu,k,p} =\frac{1}{4\pi p}\lambda_k + \bigo(p^{-2})\;
\end{equation}
as $p\rightarrow\infty$, where $\gamma_{\nu,k,p}$ is the
$k$-th eigenvalue of the Berezin transform of $W_{\nu,p}$
and $\lambda_k$ is the $k$-th eigenvalue of the Laplace-Beltrami
operator associated with the Kähler metric, for all $k\in\N$.
Note in particular that the first term of the right hand side of
equation \eqref{eq-LBnu} does not depend on the choice of the
smooth volume form $\nu$. It would be interesting to understand
the general mechanism behind this fact, in the spirit of
\cref{thm-robust}.(ii), showing that the spectral gap of
$W_{\nu,p}$ is constant up to $O(1/p^2)$ under deformations
of $\nu$.

}
\end{rem}

\section{Berezin transform and Donaldson's iterations}\label{sec-don}
In \cite{D} Donaldson, as a part of his program of developing approximate
methods for detecting canonical metrics on K\"{a}hler manifolds, discovered a remarkable
class of dynamical systems on the space of all Hermitian inner products on a given complex vector space. We shall show in this section that the linearization of such a system  at a fixed point can be identified
with the quantum channel introduced in \eqref{eq-e} above and
prove that under certain natural
assumptions, it is injective and has strictly positive
spectral gap. Using earlier results by Donaldson, we will then
deduce the main result of this section, Theorem \ref{thm-expcvcor}, stating that the iterations of this system converge exponentially fast to the fixed point.

For a complex $n$-dimensional vector space $\cV$, denote by $\Prod(\cV)$ the space
of Hermitian inner products on $\cV$. Given such a $q\in\Prod(\cV)$,
let $\cH:=(\cV,q)$ be the corresponding Hilbert space,
and define a map
\begin{equation}\label{Psiqdef}
\Phi_q: \PPP (\cV^*)\longrightarrow\cL(\cH)
\end{equation}
sending a hyperplane $H\subset\cV$, naturally seen as an element of
$\PPP(\cV^*)$ via the kernel of linear forms,
to the unique orthogonal projector $\Phi_q(H)\in\cL(\cH)$
with respect to $q$ satisfying $\Ker\Phi_q(H)=H$.
%

Let $\nu$ be a Borel measure on $\PPP (\cV^*)$, so that
$|\nu|:= \nu(\PPP (\cV^*)) < \infty$.
Following Donaldson \cite[p.\,581]{D}, we say that
$q \in \Prod(\cV)$ is $\nu$-{\it balanced} if
the operator-valued measure
\begin{equation} \label{eq-Donaldson-POVM}
dW_q(z):= n\,\Phi_q(z) \frac{d\nu(z)}{|\nu|}\;,
\end{equation}
defines an $\cL(\cH)$-valued POVM on $\PPP (\cV^*)$
as in \eqref{eq-POVM-density}.
This translates into the condition
\begin{equation} \label{eq-Donaldson-POVM-1}
n\,\int_{\PPP (\cV^*)} \Phi_q(z) \frac{d\nu(z)}{|\nu|} = \id\;.
\end{equation}

\medskip
\noindent
\begin{exam}\label{exam-general}{\rm
Consider a Hilbert space $\cH=(\cV,q)$ with $\dim_\C\cH=n$, and
let $W$ be a pure $\cL(\cH)$-valued POVM, defined as in formula
\eqref{eq-POVM-density}. Let us identify the measure
$\alpha$ on $\Om$ with a measure on $\PPP(\cV^*)$
via push-forward by the associated map
\begin{equation}\label{Fdef}
F:\Om\longrightarrow\PPP(\cV^*)\;,
\end{equation}
where $\PPP(\cV^*)$, seen as the set of hyperplanes in $\cV$,
is identified with the set of rank one projectors
in $\cS(\cH)\subset\cL(\cH)$ via their kernel as above.
It is then an immediate consequence of the definitions that
$q\in\Prod(\cV)$ is $\alpha$-balanced. }
\end{exam}

\medskip\noindent\begin{exam}\label{exam-cloudK} {\rm
As a particular case of Example \ref{exam-general},
consider the Berezin-Toeplitz POVM $W_p$
on a closed quantizable K\"{a}hler
manifold $X$ associated to a Hermitian holomorphic line bundle $L$
for $p\in\N^*$ large enough, as in Section \ref{subsec-BLB}.
The associated Hilbert space is $\cH_p=(H^0(X,L^p),q_p)$,
where $H^0(X,L^p)$ is the space of holomorphic sections of $L^p$
and $q_p\in\Prod(H^0(X,L^p))$ is the $L_2$-Hermitian product
induced by the Kähler metric. In this case, the map \eqref{Fdef}
is given by the Kodaira embedding
\begin{equation}\label{Kodemb}
F_p: X \longrightarrow \PPP(H^0(X,L^p)^*)\;,
\end{equation}
and we get as a special case of the previous example that
$q_p\in\Prod(H^0(X,L^p))$
is $\alpha_p$-balanced. Then following e.g.
\cite[Prop.\,8.3]{Fine-lectures} and by formula
\eqref{eq-alphahbar} for $\alpha_p$, the data
$(X,L^p,h^p)$ is balanced in the sense of
Definition \ref{def-balanced} if and only if the product
$q_p$ is $dv_X$-balanced.

 }\end{exam}
\medskip
\noindent
\begin{exam}\label{exam-balanced}
{\rm Let $X$ be a complex manifold together with a holomorphic
line bundle $L$ over $X$ such that the Kodaira map $F_p$
given by \eqref{Kodemb} is an embedding for $p\in\N^*$
sufficiently large, and let $\nu$ be a smooth volume form over $X$.
Then $L^p$ over $X$ is naturally identified with the pullback
by $F_p$ of the dual of the
tautological line bundle over $\PPP(H^0(X,L^p)^*)$, and
given a Hermitian inner product $q_p\in\Prod(H^0(X,L^p))$
on $H^0(X,L^p)$, we write $h_p$ for the Hermitian metric
induced on $L^p$ by the corresponding Fubini-Study metric.
Then by e.g. \cite[p.\,581]{D},
the product $q_p\in\Prod(H^0(X,L^p))$ is $\nu$-balanced
if and only if it coincides up to constant with the $L_2$-inner
product on $H^0(X,L^p)$ induced by $h_p$ and the measure $d\nu$.
On the other hand, following \cite[\S\,4.19,\S\,5.1.4]{MM07},
the last assertion of Example \ref{exam-cloudK}
holds in the same way when one replaces $dv_X$ by $d\nu$,
so that the weighted Rawnsley function
$R_{\nu,p}:X\rightarrow\R$ of Remark \ref{weightedBT} is
constant if and only if $q_p$ is $\nu$-balanced.
This shows that if $q_p$ is $\nu$-balanced,
the induced POVM \eqref{eq-Donaldson-POVM}
coincides with the Berezin-Toeplitz POVM \eqref{Wnup}
weighted by $\nu$ of Remark \ref{weightedBT}.

Donaldson proved \cite[p.\,582]{D} (see also an extensive discussion below)
that for every $p\in\N^*$ large enough, there
always exists a unique $\nu$-balanced Hermitian inner product
$q_p\in\Prod(H^0(X,L^p))$. For $p \in N^*$ large enough, consider the symplectic form
$\omega_p$ on $X$ obtained by the pull-back under the Kodaira map $F_p$
of the Fubini-Study form on $\PPP(H^0(X,L^p)^*)$ corresponding to $q_p$.
Equivalently, $-2i\pi\om_p$ is the Chern curvature of $h_p$.
By \cite[p.\,584]{D} (see also \cite{Keller09}), the sequence
$\{\frac{1}{p}\omega_p\}_p$  converges as $p\rightarrow\infty$
to the unique Kähler form $\om_\infty$ in $c_1(L)$
solving the \emph{Calabi problem} $\om^d=c\nu$, for some
$c>0$. This illustrates the role of $\nu$-balanced
products as finite-dimensional approximations of the solution
of the Calabi problem.


}
\end{exam}

\medskip

Under some natural assumptions on the measure $\nu$,  the existence
of $\nu$-balanced Hermitian inner products  was established by
Bourguignon, Li and Yau \cite{BLY94}, where they use such products
to give an upper bound
for the first eigenvalue of the Laplacian
of complex manifolds embedded in the projective space.
This generalizes the seminal work
of Hersch \cite{Her70}, where he shows that the first eigenvalue
of any metric over $S^2$ is smaller than the one of the
round metric, using the notion of balanced
product in its simplest form.

Following instead Donaldson in \cite{D}, let us associate to a measure
$\nu$ on $\PPP(\cV^*)$ the dynamical
system $\TT_\nu:\Prod(\cV)\rightarrow\Prod(\cV)$ defined
for all $q\in\Prod(\cV)$ by
\begin{equation}\label{Tnudefcan}
\TT_\nu(q):=n\,\int_{\PPP (\cV^*)} q\left(\Phi_q(z)\,\cdot\,,\,\cdot\,\right)
\,\frac{d\nu(z)}{|\nu|}\,.
\end{equation}
Using condition \eqref{eq-Donaldson-POVM-1}, we then see that
$q\in\Prod(\cV)$ is $\nu$-balanced
if and only if it is a fixed point of $\TT_\nu$.
Under mild conditions on the measure $\nu$,
Donaldson proved that for every initial condition
$q_0\in\Prod(\cV)$, the iterations $\TT_\nu^r (q_0)$
converge to such a fixed point as $r \to +\infty$,
and that this fixed point is
unique up to the action of $\R_+$ on $\Prod(\cV)$ by scalar
multiplication.

The main result of
this section is the
{\it exponential convergence} of Donaldson's iteration
process to the $\nu$-balanced product, for all initial conditions.

\begin{thm}\label{thm-expcvcor}
Suppose that the measure $\nu$ on $\PPP(\cV^*)$
is supported on a complex subvariety
$Y \subset \PPP(\cV^*)$, with $\nu$ absolutely
continuous on every irreducible component
of $Y$. Assume that
\begin{itemize}
\item[{(i)}] for any projective
subspace $\Sigma$ of $\PPP(\cV^*)$, we have
\begin{equation}
\label{eq-spade-1}
\frac{\nu(\Sigma)}{\dim \Sigma +1} < \frac{|\nu|}{n}\;;
\end{equation}
\item[{(ii)}] at least one irreducible
component of $Y$ is not contained in any proper projective
subspace of $\PPP(\cV^*)$.
\end{itemize}
Then for any $q_0\in\Prod(\cV)$,
there exists a $\nu$-balanced product $q\in\Prod(\cV)$ and constants
$C>0,\,\beta\in(0,1),$
 such that for all $r\in\N$, we have
\begin{equation}\label{expcvest}
\textup{dist}(\TT_\nu^r(q_0),q)\leq C\beta^r\;.
\end{equation}
\end{thm}

\medskip\noindent

Note that if $Y$ is irreducible, assumptions $(i)$ and $(ii)$
are satisfied as soon as $Y$ is not contained in a proper
projective subspace of $\PPP(\cV^*)$. Thus Theorem \ref{thm-expcvcor}
applies in particular to the important case
of Example \ref{exam-balanced}, where $\nu$ is induced by a smooth
volume form over a complex
manifold $Y$ embedded in a projective space via Kodaira embedding.
Conversely, if the whole variety $Y$
(in contrast with its irreducible components)
lies in a proper projective subspace of $\PPP(\cV^*)$,
then there exists $u\in\cV$ such
that $\Phi_q(z)\,u=0$ for all $z\in Y$, contradicting
condition \eqref{eq-Donaldson-POVM-1}, so that there does not
exist any $\nu$-balanced Hermitian product.

The proof of Theorem \ref{thm-expcvcor} will rely on the
Propositions
\ref{Psicv}, \ref{prop-D2} and \ref{prop-D3} below.
In particular, Proposition \ref{Psicv} generalizes the result
of Donaldson in \cite[p.\,581]{D}, which essentially
states that the iterations $\TT_\nu^r(q_0)$
converge to a $\nu$-balanced product as $r\rightarrow+\infty$
for all $q_0\in\Prod(\cV)$
if either $Y$ is a complex variety which is not contained in any proper projective subspace, or $Y$
is a finite collection of points satisfying $(i)$.
Specifically, Donaldson's
assumption 2 in \cite[p.\,581]{D} is precisely assumption $(i)$
in the case $\dim Y=0$, and Donaldson's
assumption 1 in \cite[p.\,581]{D} is satisfied in the
case $Y$ is a complex variety which is not contained in any proper projective subspace, but do not imply assumption $(ii)$ in general.
The proof of Proposition \ref{Psicv}
follows closely the lines of \cite[p.\,581]{D}.

On the other hand, the role of Propositions \ref{prop-D2}
and \ref{prop-D3} in the proof of Theorem \ref{thm-expcvcor} is based
on the key observation, which is a reformulation of
\cite[p.\,609]{D}, that the linearization of $\TT_\nu$
at a fixed point $q\in\Prod(\cV)$ coincides
with the quantum channel \eqref{eq-e} of the POVM
\eqref{eq-Donaldson-POVM}
associated with $q$.

To see this, let us first choose a base point
$q_0\in\Prod(\cV)$ and identify $(\cV,q_0)$ with
$(\C^n, \langle\cdot,\cdot\rangle)$, where $\langle z,w\rangle = \sum_j z_j\bar{w}_j$.
Writing $\cL(\C^n)_+$ for the set of positive
Hermitian $n \times n$ matrices, this identifies
$G \in \cL(\C^n)_+$ with $q(\cdot,\cdot):= \<G\cdot,\cdot\>\in\Prod(\cV)$.
Next, identify $[z] \in \C P^{n-1}$ with the hyperplane
$$\{w\;:\; \< z,w\>= q(G^{-1}z,w) = 0\}\;.$$
From the definition \eqref{Psiqdef} of $\Phi_q$
we have
$$\Phi_q([z])\xi = \frac{\<G\xi,G^{-1}z \>}{\<GG^{-1}z,G^{-1}z\>} G^{-1}z
=G^{-1} \left(\frac{\<\xi,z\>}{|z|^2}z\right) \cdot \frac{|z|^2}{\<G^{-1}z,z\>}
=G^{-1}\Pi_z\xi \cdot \frac{|z|^2}{\<G^{-1}z,z\> }    \;,$$
where $\Pi_z$ denotes the orthogonal projector
with respect to $\<\cdot,\cdot\>$
to the line generated by $z \in \C^n\backslash\{0\}$.
Thus,
$$q(\Phi_q(z)\xi,\xi) = \<\Pi_z\xi,\xi\> \cdot \frac{|z|^2}{\<G^{-1}z,z\> }    \;.$$
Therefore, we can reformulate the definition \eqref{Tnudefcan}
of $\TT_\nu$ in coordinates by the formula
\begin{equation}\label{Tnudef}
\TT_\nu(G)=n\,\int_{\C P^{n-1}} \Pi_z\,\frac{|z|^2}
{\langle G^{-1} z,z\rangle }\,\frac{d\nu(z)}{|\nu|}\;.
\end{equation}

%
%
%
%
%
%
%
%
%
%

Recall that the tangent space of
$\Prod(\cV)$ at any $q\in\Prod(\cV)$ is canonically
identified with the space of Hermitian operators
$\cL(\cH)$ of $\cH:=(\cV,q)$. Then if $q\in\Prod(\cV)$
is $\nu$-balanced, so that $\TT_\nu(q)=q$, the differential
$D_q\TT_\nu$ of $\TT_\nu$ at $q$ acts on $\cL(\cH)$.
\medskip
\noindent
\begin{lemma}\label{prop-D1} For any
$\nu$-balanced Hermitian product $q\in\Prod(\cV)$,
the differential of $\TT_\nu$ at $q$ satisfies $D_q\TT_\nu = \cE_q$,
where $\cE_q$ is the quantum channel \eqref{eq-e} of the
associated POVM $W_q$ defined by \eqref{eq-Donaldson-POVM}.
\end{lemma}
\begin{proof}
Let $q\in\Prod(\cV)$ be a $\nu$-balanced Hermitian product, and
let us identify $(\cV,q)$ with $(\C^n,\langle\cdot,\cdot\rangle)$
as above.
Let $G(t) \in \cL(\C^n)_+$ be a path such that $G(0)=\id$.
Abbreviating $\dot{G}:=\dot{G}(0)$ and using formula \eqref{Tnudef},
we get
$$\frac{d}{dt}\Big |_{t=0} \TT_\nu(G(t)) = n\,\int_{\C P^n} \Pi_z \frac{\langle\dot{G}z,z\rangle}{|z|^2}\,\frac{d\nu(z)}{|\nu|}\;.$$
Recall that $((\cdot,\cdot))$ denotes the natural scalar product
on $\cL(\cH)$.
Then noticing that
$\langle\dot{G}z,z\rangle/|z|^2= ((\dot{G},\Pi_z))$,
we get
$$\frac{d}{dt}\Big |_{t=0} \TT_\nu(G(t)) =
n \ \int_{\C P^n} \Pi_z ((\dot{G},\Pi_z))
\,\frac {d\nu(z)}{|\nu|}\;,$$
which is precisely formula \eqref{eq-e} for the quantum channel
associated to $W_q$ defined by \eqref{eq-Donaldson-POVM},
as we have $\Phi_q(z)=\Pi_z$ for all $z\in\C^n\backslash\{0\}$
in the identification of
$(\cV,q)$ with $(\C^n,\<\cdot,\cdot\>)$. This concludes the proof.\end{proof}

\medskip

Recall that the quantum channel $\cE_q:\cL(\cH)\to\cL(\cH)$
satisfies $\cE_q(\id)=\id$, and that its \emph{spectral gap}
is the quantity $\gamma=1-\lambda_1$, where
\begin{equation}
1=\lambda_0\geq\lambda_1\geq\lambda_2\geq\cdots\geq 0
\end{equation}
is the decreasing sequence of eigenvalues of $\cE_q$.
%
%
%
%
Then Proposition \ref{prop-D2} establishes the positivity
of the spectral gap of $\cE_q=D_q\TT_\nu$
under assumption $(i)$, showing that
Donaldson's prediction in \cite[\S\,4.1]{D} on the largest
eigenvalue of the linearization of $\TT_\nu$ at a $\nu$-balanced
product
in fact holds for general projective smooth
manifolds $Y$, as assumption $(i)$ is automatically satisfied
as soon as $Y$ is not contained in any proper projective
subspace.

Proposition \ref{prop-D3}
shows the invertibility of $\cE_q=D_q\TT_\nu$ under assumption $(ii)$.
This is a key assumption in the classical Grobman-Hartman theorem,
which we use in Theorem \ref{thm-expcvcor} to show that the iterations
of the dynamical system $\TT_\nu$ converge exponentially fast
to a fixed point. As assumption $(ii)$ is automatically satisfied
for a projective smooth manifold $Y$ not contained in any proper
projective subspace, this strengthens
Donaldson's prediction in \cite[\S\,4.1]{D} on the asymptotic
rate of convergence of the dynamical system $\TT_\nu$. Namely,
with only the positivity of the spectral gap, we expect
that the rate of convergence is exponentially fast
for almost all initial conditions, while Theorem \ref{thm-expcvcor}
shows that it actually holds for all initial conditions.

\medskip
\noindent
\begin{prop}\label{Psicv}
Assume that assumption $(i)$ of Theorem \ref{thm-expcvcor} holds.
Then for any $q_0\in\Prod(\cV)$, the iterations
$\TT_\nu^r (q_0)$ converge to a fixed point as
$r\to+\infty$, unique up to the action of $\R_+$ by scalar
multiplication.
\end{prop}
\begin{proof}
Fix $q_0\in\Prod(\cV)$, and identify $(\cV,q_0)$ with
$(\C^n,\<\cdot,\cdot\>)$, where $\<\cdot,\cdot\>$ denotes
the canonical Hermitian product of $\C^n$.
Recall that $\cL(\C^n)_+$ denotes the set of positive Hermitian
$n\times n$ matrices.
Following \cite[p.\,582]{D}, for any $[z]\in\C P^{n-1}$, let
$z\in\C^n$ be a lift of norm $1$, and for any $G\in\cL(\C^n)_+$,
set
\begin{equation}\label{psizdef}
\psi_{[z]}(G):=\log\<G^{-1}z,z\>+\frac{1}{n}\log\det G.
\end{equation}
This quantity does not depend on the choice of a lift
of $[z]\in\C P^{n-1}$ of norm $1$, and the second term makes
it invariant under multiplication of $G$ by a positive scalar.
Given a Borel measure $\nu$ on $\C P^{n-1}$, we then define a
functional on $\cL(\C^n)_+$ by the formula
\begin{equation}\label{Psinudef}
\Psi_\nu(G)=\int_{\C P^{n-1}}\psi_{[z]}(G)\,d\nu([z]),
\end{equation}
for any $G\in\cL(\C^n)_+$.
Using formula \eqref{Tnudef},
we see that $G\in\cL(\C^n)_+$
is a critical point of $\Psi_\nu$ if and only if it is
a fixed point of $\TT_\nu$. Thus to show the existence and
unicity of such a fixed point up to the action of $\R_+$,
we can restrict $\Psi_\nu$
to the space $\cL(\C^n)_+^1$ of positive Hermitian matrices of
determinant $1$, and it suffices to show that $\Psi_\nu$
is strictly convex and proper along any geodesic of $\cL(\C^n)_+^1$
for its natural Riemannian metric as a symmetric space.
In fact, any strictly convex and proper function over $\R$ has
a unique absolute minimum, which is also its unique critical point.
Now as two points can always be joined by a geodesic, we
conclude in that case
that a fixed point of $\TT_\nu$ on $\cL(\C^n)_+^1$ coincide
with a minimum of $\Psi_\nu$, which exists and is unique.

Recall that the structure of symmetric space on $\cL(\C^n)_+^1$
is given by the map
\begin{equation}
\begin{split}
\textup{SL}_n(\C)&\longrightarrow\cL(\C^n)_+^1\\
G&\longmapsto \sqrt{G^*G}\;,
\end{split}
\end{equation}
which realizes $\cL(\C^n)_+^1$ as the quotient of
the special linear group $\textup{SL}_n(\C)$ by the special unitary
group $\textup{SU}(n)$. The usual scalar product $((\cdot,\cdot))$
on the space of $n \times n$ matrices induces a Riemannian
metric on $\cL(\C^n)_+^1$ through the identification of its
tangent space at any point with the space of
traceless matrices. By general theory of symmetric spaces,
geodesics are simply the images of $1$-parameter groups of
$\textup{SL}_n(\C)$ through the above map,
so that up to the action of
$\textup{SU}(n)$ by conjugation, they are of the form
$G_t\in\cL(\C^n)_+^1$, with
\begin{equation}
G_t=\diag(e^{\lambda_1 t},e^{\lambda_2 t},\cdots,e^{\lambda_n t})\;,
\end{equation}
for all $t\in\R$,
where $\lambda_1\geq\lambda_2\geq\cdots\geq\lambda_n$ satisfy
$\sum_{j=1}^n\lambda_j=0$. Now if $\nu$ satisfies assumption $(i)$
of Theorem \ref{thm-expcvcor}, its pullback by the action of a unitary
matrix also satisfies this assumption, and thus we are reduced
to show strict convexity and properness of
\begin{equation}
t\longmapsto\Psi_\nu(G_t)=\int_{\C P^{n-1}}\log\left(\sum_{j=1}^ne^{\lambda_j t}
|z_j|^2\right)\,d\nu([z]),\quad t\in\R\;.
\end{equation}
Now convexity follows from a direct computation, with
strict convexity as long as the total mass of $\nu$ is not contained
in any projective subspace of $\C P^{n-1}$, which is a
straightforward consequence of assumption $(i)$.

Let us now show properness, i.e. that $\Psi_\nu(G_t)\to+\infty$
when $t\to\pm\infty$. By considering the geodesic going to the
opposite direction, it suffices to show it when $t\to+\infty$. Consider an irreducible component
$Z\subset Y$, and let $k\leq n$ be the largest integer such that
$Z$ is contained in the projective subspace
\begin{equation}\label{Sigmak}
\Sigma_k:=\{[0:\cdots:0:z_k:\cdot:z_n]\in\C P^{n-1}\}\subset\C P^{n-1}
\end{equation}
As $\nu$ is absolutely continuous over the smooth part of $Z$,
this means in particular that the function $\log|z_k|^2$
restricted to $Z$ is integrable with respect to $\nu$.
We thus get a constant $C_Z>0$ such that
\begin{equation}
\begin{split}
\int_Z\log\left(\sum_{j=1}^ne^{\lambda_j t}
|z_j|^2\right)\,d\nu([z])&\geq\int_Z\log\left(e^{\lambda_k t}
|z_j|^2\right)\,d\nu([z])\\
&\geq\lambda_k t\nu(Z)-C_Z\;.
\end{split}
\end{equation}
For any $k\leq n$, write $\nu_k>0$ for the total mass of
the irreducible components of $Y$ for which $k$ is the largest
integer such that they are not contained in $\Sigma_k$ as above.
We then get a constant $C_Y>0$ such that
\begin{equation}
\Psi_\nu(G_t)\geq t\sum_{j=1}^n\lambda_j\nu_j-C_Y\;.
\end{equation}
We are thus reduced to show that $\sum_{j=1}^n\lambda_j\nu_j>0$.
Notice now that assumption $(i)$ implies
\begin{equation}
\sum_{j=k}^n\nu_j<\frac{n-k}{n}\sum_{j=1}^n\nu_j,
\quad\text{for all}~~1\leq k\leq n\;.
\end{equation}
Using $\lambda_1\geq\lambda_2\geq\cdots\geq\lambda_n$ and
$\sum_{j=1}^n\lambda_j=0$, we then get
\begin{equation}
\begin{split}
\sum_{j=1}^n\lambda_j\nu_j&=\lambda_0\sum_{j=1}^n\nu_j+
\sum_{k=1}^n(\lambda_k-\lambda_{k-1})\sum_{j=k}^n\nu_j\\
&>\left(\lambda_0+\sum_{k=1}^n\frac{n-k}{n}
(\lambda_k-\lambda_{k-1})\right)\sum_{j=1}^n\nu_j\\
&>\left(\frac{1}{n}
\sum_{k=1}^n\lambda_k\right)\sum_{j=1}^n\nu_j=0\;.
\end{split}
\end{equation}
This implies properness.

Let us now show the convergence of iterations of $\TT_\nu$ to
a fixed point. We will first show that $\TT_\nu$ decreases
$\Psi_\nu$, so that iterations have an accumulation point
by properness, and we will then show that this accumulation point
is in fact a fixed point.
First note that for any $G\in\cL(\C^n)_+$, using the fact
that projectors are of trace $1$,
formula \eqref{eq-Donaldson-POVM-1}, together
with \eqref{Tnudef}, gives
$\tr\left[\TT_\nu(G)G^{-1}\right]=n$.
Using the strict concavity of the logarithm, we thus get
\begin{equation}\label{logdet}
\begin{split}
\frac{1}{n}\log\det\left(\TT_\nu(G)\right)
-\frac{1}{n}\log\det\left(G\right)&=
\frac{1}{n}\log\det\left(\TT_\nu(G)G^{-1}\right)\\
&\leq\log\left(\frac{\tr\left[\TT_\nu(G)G^{-1}\right]}{n}\right)
=0\;,
\end{split}
\end{equation}
with equality if and only if $\TT_\nu(G)G^{-1}=\id$.
Thus to show that $\Psi_\nu(\TT_\nu(G))\leq\Psi_\nu(G)$,
by definition \eqref{Psinudef} of $\Psi_\nu$,
we only need to show that $\TT_\nu$ decreases
the integral against $\nu$
of the first term of formula \eqref{psizdef}.
Again by concavity of the logarithm,
we get
\begin{equation}\label{intlog}
\begin{split}
\int_{\C P^{n-1}}\log\<\TT_\nu(G)^{-1}z,z\>\,d\nu([z])-
\int_{\C P^{n-1}}&\log\<G^{-1}z,z\>\,d\nu([z])\\
&\leq\log\left(\int_{\C P^{n-1}}\frac{\<\TT_\nu(G)^{-1}z,z\>}
{\<G^{-1}z,z\>}\,d\nu([z])
\right)\\
&\leq\log\left(\frac{1}{n}
\tr\left[\TT_\nu(G)\TT_\nu(G)^{-1}\right]\right)=0\;,
\end{split}
\end{equation}
where we used formula \eqref{Tnudef} for $\TT_\nu(G)$ together
with the fact that $\<Az,z\>=|z|^2\Tr\left[\Pi_zA\right]$ for all
$z\in\C^n\backslash\{0\}$ and $A\in\End(\C^n)$.
Equations \eqref{logdet} and \eqref{intlog}, together with the
definition of $\Psi_\nu$ given by formulas \eqref{psizdef} and
\eqref{Psinudef}, show that
$\Psi_\nu(\TT_\nu(G))\leq\Psi_\nu(G)$ for all
$G\in\cL(\C^n)_+$.

To conclude, note first that properness
over $\cL(\C^n)_+^1$ and invariance
under the action of $\R_+$ implies that $\Psi_\nu$
is bounded from below over the whole $\cL(\C^n)_+$.
Thus for any $G_0\in\cL(\C^n)_+$,
we get that the decreasing sequence
$\{\Psi_\nu(\TT_\nu^r(G_0))\}_{r\in\N}$
converges to its lower bound. As both terms in the definition
of $\Psi_\nu$ are decreasing under iterations of $\TT_\nu$
by \eqref{logdet} and \eqref{intlog}, we then deduce that
$\{\log\det(\TT_\nu^r(G_0))\}_{r\in\N}$, thus
also $\{\det(\TT_\nu^r(G_0))\}_{r\in\N}$, are bounded in $\R$,
and that
\begin{equation}\label{logdetto0}
\frac{1}{n}\log\det\left(
\TT_\nu^{r+1}(G_0)\TT_\nu^r(G_0)^{-1}\right)
\longrightarrow 0,\quad\text{as}~~r\to+\infty\;.
\end{equation}
Now from properness of $\Psi_\nu$ over $\cL(\C^n)_+^1$
and boundedness in $\R$ of the sequences
$\{\Psi_\nu(\TT_\nu^r(G_0))\}_{r\in\N}$ and
$\{\det(\TT_\nu^r(G_0))\}_{r\in\N}$,
we get that the sequence $\{\TT_\nu^r(G_0)\}_{r\in\N}$
admits an accumulation point
$G_\infty\in\cL(\C^n)_+$.
On the other hand,
by strict concavity of the logarithm, formula \eqref{logdetto0}
and the equality case in
formula \eqref{logdet} imply
\begin{equation}
\TT_\nu^{r+1}(G_0)\TT_\nu^r(G_0)^{-1}
\longrightarrow\id,\quad\text{as}~~r\to+\infty\;.
\end{equation}
We thus get that $G_\infty\in\cL(\C^n)_+$ is the unique
accumulation point, and satisfies
$\TT_\nu(G_\infty)=G_\infty$. This concludes the proof.
\end{proof}

In the following Proposition, we use the result that
a fixed point of $\TT_\nu$ exists as soon as $\nu$ satisfies
assumption $(i)$, which was proved in the previous Proposition.

\medskip
\noindent
\begin{prop}\label{prop-D2}
Assume that assumption $(i)$ holds.
Then for any $\nu$-balanced product $q\in\Prod(\cV)$,
the associated quantum
channel $\cE_q$ as in \cref{prop-D1}
has positive spectral gap.
\end{prop}
\begin{proof}
Let $q\in\Prod(\cV)$ be a $\nu$-balanced product,
and identify $(\cV,q)$ with $(\C^n,\<\cdot,\cdot\>)$,
so that $\Phi_q(z)=\Pi_z$ for all $z\in\C^n\backslash\{0\}$
in the definition \eqref{eq-Donaldson-POVM} of $W_q$,
where $\Pi_z$ is the orthogonal projector on $[z]$ with respect
to $\<\cdot,\cdot\>$.
Assume that $\nu$
satisfies assumption $(i)$ of Theorem
\ref{thm-expcvcor}, and normalize it by setting  $\alpha:= \nu/|\nu|$.
For any $z\in\C^n\backslash\{0\}$, we denote by $[z]$
its class in $\C P^{n-1}$.
For any $z,\,w\in\C^n\backslash\{0\}$, we write
\begin{equation}
\cB_q([z],[w])= n\frac{|\langle z, w \rangle |^2}
{|z|^2|w|^2}
\end{equation}
for the Schwartz kernel with respect to $\alpha$ of the
Berezin transform \eqref{eq-b} on $L_2(\C P^{n-1},\nu)$ associated
with $W_q$.
Let $Y_1,\dots, Y_k$ be
the irreducible components
of $Y$. Since $(z,w) \mapsto \langle z, w\rangle$ is holomorphic in $z$ and anti-holomorphic in $w$,
for every $i,\,j\leq k$, we get that
\begin{itemize}
\item[{(a)}] either $\cB_q([z],[w]) = 0$
for all $([z],[w])\in Y_i \times Y_j\;,$
\item[{(b)}] or $\cB_q([z],[w])\neq 0$ for almost all
$([z],[w])\in Y_i \times Y_j\;.$
\end{itemize}
Consider a graph $\Gamma$ with vertices $1,\dots, k$,
where $i,j$ are connected by an edge whenever $(b)$ occurs
of $Y_i\times Y_j$.
In particular, each $i$ is connected by an edge to itself.

Recall that
$$\int \cB_q(x,y) d\alpha(y) = \int \cB_q(x,y) d\alpha(x)=1\;.$$
Using the Schur test as in formula \eqref{Schurtest} above,
we apply Cauchy-Schwarz inequality on the formula
\begin{equation}\label{CS1/2}
\int\cB_q(x,y)\phi(y)d\alpha(y) =
\int \cB_q(x,y)^{1/2}\, \cB_q(x,y)^{1/2}\phi(y) d\alpha(y)\;,
\end{equation}
to get for any $\phi\in L^2(\C P^{n-1},\nu)$,
\begin{equation}
\begin{split}
\|\cB_q\phi\|^2_{L_2} &= \int \left (\int \cB_q(x,y)\phi(y)
d\alpha(y)\right) ^2 d\alpha(x)\\
&\leq    \int \left ( \int \cB_q(x,y) d\alpha(y) \cdot \int \cB_q(x,y)
\phi^2(y)d\alpha(y) \right) d\alpha(x)  =\|\phi\|_{L_2}\;.
\end{split}
\end{equation}
In particular, the equality $\cB_q\phi=\phi$ can hold only
if the inequality above is an equality, and by the equality
case of Cauchy-Schwarz inequality,
this implies that for
$\alpha$-almost all $x$, there exists $c \neq 0$ such that
$c\cB_q(x,y)^{1/2}_q= \cB(x,y)^{1/2}\phi(y)$ for $\alpha$-almost
all $X$. In terms of the graph
defined in the previous step, this yields that $\phi$ is constant on every subset of the form
$\bigcup_{j \in \text{star}(i)} Y_j,$ where $i=1,\dots ,k$.
Thus if $\phi$ is a non constant function satisfying $\cB_q\phi=\phi$,
it follows that $\Gamma$ is disconnected. Denote by $\Gamma_i$, $i=1,\dots,k$ the connected components,
and put $Z_i= \bigcup_{j \in \Gamma_i} Y_j$.

Assuming that there exists a non-constant $\phi$ satisfying
$\cB_q\phi=\phi$ as above, we will show that
assumption $(i)$ can not hold.
Recall that we work with the POVM $dW_q(x) = n \Pi_x d\alpha(x)$,
where $\Pi_x$ is the orthogonal
projector to the line $x \in \C P^{n-1}$ with respect to
$\<\cdot,\cdot\>$. With this notation,
$\cB_q(x,y)=0$ yields $\Pi_x\Pi_y = 0$.
Write $P = W_q(Z_1)$ and $P' =  W_q(Z_2 \cup \dots \cup Z_k)$.
It follows that
$P+P'=\id$ and $PP' = 0$. Thus $P$ is an orthogonal projector whose image is a proper projective subspace $\Sigma$ of $\C P^{n-1}$ of dimension $m-1$, with
$$m = \tr\,\left[P\right] =  \tr\,\left[W_q(Z_1)\right]
= n\,\alpha(Z_1)
= n\,\frac{\nu(Z_1)}{|\nu|}\;.$$
Observe also that if $Pz=0$, we get
$$\int_{Z_1} \langle \Pi_x z,z \rangle\, d\nu(x) =0\;,$$
and hence $\langle \Pi_x z,z\rangle = 0$ for $\nu$-almost all $x$.
Since $\nu$ is absolutely continuous on each irreducible component of $X$, it follows
that $x$ is orthogonal to $z$ for all $x \in Z_1$, and hence $Z_1 \subset \Sigma$. We conclude that
$$\frac{\nu(\Sigma)}{m} \geq \frac{\nu(Z_1)}{m}= \frac{|\nu|}{n}\;,$$
so that assumption $(i)$ does not hold.
\end{proof}

The proof of this last Proposition is a variation on the theme of \cite[Proposition 4.1]{BMS}.

\medskip
\noindent
\begin{prop}\label{prop-D3}
Assume that assumption $(ii)$ holds. Then for any $\nu$-balanced
product $q\in\Prod(\cV)$,
the associated quantum channel
$\cE_q$ as in \cref{prop-D1}
is invertible.
\end{prop}

\medskip\noindent

\begin{proof} Let $q\in\Prod(\cV)$ be $\nu$-balanced, and
identify $(\cV,q)$ with $(\C^n,\<\cdot,\cdot\>)$.
For any $z\in\C^n\backslash\{0\}$, we denote by $[z]$
its class in $\C P^{n-1}$.

Denote by $\tilde{Y}$ the cone of $Y$ in $\C^n$.
Assume on the contrary that an Hermitian matrix $A \neq 0$ lies in the kernel of $\cE_q$, and set
$$F_A([z],[w]):= \frac{\langle Az, w \rangle}{|z|\cdot |w|}\;.$$
Since $\cE_q= n^{-1} TT^*$ and $T^* (A) ([z]) = F_A([z],[z])$ by
the results of Section  \ref{prel},
we have $F_A([z],[z]) =0$ for all $[z] \in Y$. Noticing that the function $(z,w) \mapsto \langle Az, w \rangle$
is holomorphic in $z$ and anti-holomorphic in $w$ and that it vanishes on the diagonal of $\tilde{Y} \times \tilde{Y}$, we conclude that $F$ vanishes on $Z \times Z$ for every irreducible component $Z$ of $Y$.

Pick any irreducible component $Z$. If it fully lies in
$\text{Ker}\, A$, we have
that $Z$ is contained in a proper projective subspace. Otherwise, pick $[u] \in Z$ so that $Au \neq 0$.
We thus proved that any other $[z] \in Z$
satisfies a linear equation $\langle z, Au \rangle = 0$, meaning that $Z$
lies in a proper projective subspace. This is in contradiction
with assumption $(ii)$.
\end{proof}



\medskip
Using the Propositions above, we are then ready to prove
Theorem \ref{thm-expcvcor}.

\medskip\noindent{\bf Proof of Theorem \ref{thm-expcvcor}:}
Suppose that $\nu$ satisfies the assumptions $(i)$ and $(ii)$,
and fix $q_0\in\Prod(\cV)$. By Proposition \ref{Psicv},
the iterations $\TT^r(q_0)$ converge to a fixed point $q\in\Prod(\cV)$
as $r \to +\infty$, so that we can use it to identify
$(\cV,q)$ with $(\C^n,\<\cdot,\cdot\>)$.
Identify diffeomorphically
$\cL(\C^n)_+$ with $\cL(\C^n)_+^1 \times \R_+$ via the map
$$\Theta: G \longmapsto\left(\cD(G),\det(G)\right),\;\;\text{where}\;\; \cD(G):=\frac{G}{\det(G)^{1/n}}\;.$$
Then for every $r \in \N$,
\begin{equation}
\label{eq-newcoo}
\Theta \TT^r_\nu \Theta^{-1} (G,g) =\left(\cD(\TT^r_\nu(G)), g \cdot \det\TT^r_\nu(G)\right)\;.
\end{equation}
Recall that by Lemma \ref{prop-D1}, the differential of
$\TT_\nu$ at $q$ coincides with the quantum channel $\cE_q$,
and recall that $q\in\Prod(\cV)$ is sent to the
identity $\id\in\cL(\C^n)_+$ in the identification of $(\cV,q)$
with $(\C^n,\<\cdot,\cdot\>)$.
Since $\cL(\C^n)^1_+$ is a slice of the $\R_+$-action and
$\TT_\nu$ is $\R_+$-equivariant,
the differential of $\cD \circ \TT_\nu$ equals to the restriction of
$\cE_q$ to the tangent space $T_\id\cL(\C^n)_+^1$, which
consists of all trace $0$ Hermitian matrices. Then by Propositions \ref{prop-D2} and \ref{prop-D3}, the spectrum
of this differential is contained in $(0,1)$, so that
$\cD \circ \TT_\nu$ is a
local diffeomorphism of $\cL(\C^n)_+^1$ in a neighborhood its hyperbolic fixed point $\id$, and conjugate through a local
homeomorphism to its linearization at $\id$ by the classical
Hartman-Grobman theorem. In particular, taking $\beta\in(0,1)$
as the
largest eigenvalue of $\cE$ in $(0,1)$, we get a constant $C>0$
such that
\begin{equation}\label{Tnudetexpcv}
\textup{dist}
\left(\cD (\TT_\nu^r(G_0)),\id
\right)\leq C\beta^r\;,\quad\text{for all}\; r\in\N\;,
\end{equation}
where $G_0\in\cL(\C^n)_+$ denotes the image of $q_0\in\Prod(\cV)$
in the identification of $(\cV,q)$ with $(\C^n\<\cdot,\cdot\>)$.
By \eqref{eq-newcoo}, in order to complete the proof of the exponential convergence of the orbit of $G_0$
to $\id$, we need to show that for $r$ large enough
\begin{equation}\label{detcv}
\left|\det \TT^r_\nu(G_0)-1\right|<C\beta^r\;.
\end{equation}
To this end recall that the functional $\Psi_\nu$ of the proof of \cref{Psicv}
is decreasing under iterations of $\TT_\nu$
and invariant with respect to the action
of $\R_+$ by multiplication. By \eqref{Tnudetexpcv}
and the differentiability of $\Psi_\nu$ at $\id$, there exists a constant
$C>0$ such that
\begin{equation}
0\leq\Psi_\nu(\TT^r_\nu(G_0))-\Psi_\nu(\id)\leq C\beta^r
\;.
\end{equation}
Now as both \eqref{logdet} and \eqref{intlog} are non-positive
and as $\TT^r_\nu(G_0)\to\id$ as $r\to +\infty$, recalling the
definition \eqref{psizdef}-\eqref{Psinudef} of $\Psi_\nu$
we deduce in particular that
\begin{equation}
0\leq\log\det(\TT^r_\nu(G_0))\leq C\beta^r\;.
\end{equation}
Since for $x$ close to $1$, we have $2|\log x | \geq |1-x|$,
this yields \eqref{detcv}. The proof is complete.
\qed

\medskip\noindent

\begin{rem}\label{thmquantnu}
{\rm
Consider the setting of Example \ref{exam-balanced} above
for all $p\in\N^*$ large enough,
where $q_p\in\Prod(H^0(X,L^p))$ is the unique $\nu$-balanced
product and $h_p$ the induced Fubini-Study metric on $L^p$
over $X$, with Chern curvature $-2i\pi\om_p$. Recall
that in that case, the induced POVM \eqref{eq-Donaldson-POVM}
coincides with the weighted Berezin-Toeplitz POVM \eqref{Wnup}
of Remark \ref{weightedBT}.
Then using a refined version of Theorem \ref{thm-quant}
and \cref{prop-D1},
one can show that the exponential convergence rate $\beta_p\geq 0$
of Donaldson's iterations in \cref{thm-expcvcor} satisfies
as $p\rightarrow\infty$ the estimate
\begin{equation}\label{thmquantunif}
\beta_p=\frac{\lambda_1(\om_\infty)}{4\pi p}+o(p^{-1})\,,
\end{equation}
where $\lambda_1(\om_\infty)$ is the first eigenvalue of the
Laplace-Beltrami operator associated with the metric induced by
the unique Kähler form $\om_\infty$ in $c_1(L)$ solving the
Calabi problem $\om^d=c\nu$ for some $c>0$. This follows
from the estimate \eqref{eq-LBnu} on the spectral gap
of the weighted Berezin transform,
together with the uniformity on the metric in the
estimates of \cite[Th.\,4.18']{DLM06} and the fact that
the sequence $\{\frac{1}{p}\om_p\}$ converges to $\om_\infty$
as $p\rightarrow\infty$.
This complements a result of Keller
in \cite[Prop.\,4.7]{Keller09}.
}
\end{rem}

\section{POVMs and geometry of measures} \label{sec-bestfit}
Assume that we are given an $\cL(\Hilb)$-valued POVM on $\Omega$ satisfying equation \eqref{eq-POVM-density}, i.e., of the form $dW = n\,F\,d\alpha$ for some $F:\Omega \to \cS(\Hilb)$. In this section we discuss spectral properties of the Berezin transform associated to $W$ in terms of the geometry of the measure
\begin{equation}\label{eq-mmm}
\sigma_W:= F_*\alpha
\end{equation}
on $\cS(\Hilb)$, focusing on its multi-scale features, and on stability of the spectral gap under perturbations of the measure.  Recall that for pure POVMs we have encountered  measure \eqref{eq-mmm} in Example \ref{exam-general}.

\medskip

Write $\cV \subset \cL(\Hilb)$ for the affine subspace consisting of all trace $1$ operators,
$\text{dist}$ for the distance on $\cV$ associated to the scalar product $((A,B))=\tr(AB)$ on $\cL(\Hilb)$.
Given a compactly supported probability measure
$\sigma$ on $\cV$, introduce the following objects:
\begin{itemize}
\item the center of mass $C(\sigma)= \int_\cV v d\sigma(v)$;
\item the mean squared distance from the origin,
$$I(\sigma)= \int_\cV \text{dist}(C,v)^2 d\sigma(v)\;;$$
\item the mean squared distance to the best fitting line
$$J(\sigma) = \inf_\ell \int_\cV \text{dist}(v,\ell)^2 d\sigma(v)\;,$$
where the infimum is taken over all affine lines $\ell \subset \cV$.
\end{itemize}
The infimum in the definition of $J$ is attained at the (not necessarily unique) {\it best fitting line}
which is known to pass through the center of mass $C$ (Pearson, 1901; see \cite[p.\,188]{Fare} for
a historical account). \footnote{The problem of finding $J$ and the corresponding minimizer $\ell$
appears in the literature under several different names including ``total least squares" and
``orthogonal regression".}

Observe that the center of mass $C(\sigma_W)$ for the measure $\sigma_W$ given by \eqref{eq-mmm} coincides with the maximally mixed state $\frac{1}{n}\id$.

\medskip
\noindent
\begin{thm}\label{prop-best-fit} The spectral gap $\gamma(W)$ depends
only on the push-forward measure $\sigma_W$ on $\cS(\Hilb)$:
$$\gamma(W) = 1 - n(I(\sigma_W)-J(\sigma_W))\;.$$
\end{thm}

\begin{proof} Let $\ell \subset \cV$ be any line passing through the center of mass $\frac{1}{n}\id$
generated by a trace zero unit vector $A \in \cL(\Hilb)$. For a point $B \in \cV$ we have
$$\text{dist}(B,\ell)^2= ((B-\frac{1}{n}\id,B-\frac{1}{n}\id))- ((B-\frac{1}{n}\id,A))^2\;.$$
Integrating over $\sigma_W$ and taking infimum over $\ell$ we get that
\begin{equation}
\label{eq-IJK-vsp}
J(\sigma_W) = I(\sigma_W) - K\;,
\end{equation}
with
\begin{equation}
\label{eq-Ksup}
K= \sup_{\substack{\tr(A)=0 \\ \tr(A^2)=1}} \int_\cV ((B,A))^2 dF_*\alpha(B)\;.
\end{equation}
The latter integral can be rewritten as
\begin{equation}
\label{eq-vsp-e}
\int_\Omega ((F(s),A))^2 d\alpha(s) = n^{-1}((\cE(A), A))\;,
\end{equation}
so by definition $K = n^{-1}\gamma_1= n^{-1}(1-\gamma(W))$. Substituting this into \eqref{eq-IJK-vsp}, we deduce
the theorem.
\end{proof}

\medskip
\begin{rem}\label{rem-eigenv}{\rm Observe that the supremum in \eqref{eq-Ksup} is attained at a unit vector $A$
generating the best fitting line. By \eqref{eq-vsp-e}, $A$ is an eigenvector of $\cE$ with the eigenvalue $\gamma_1$.
}
\end{rem}

\medskip
\begin{exam}{\rm For a pure POVM $W$, i.e. when $F$ is a one-to-one map from $\Omega$ to the set of rank-one projectors,
$$\text{dist}(C,F(s))^2
= \tr\left[\left(\frac{1}{n}\id -F(s)\right)^2\,\right]
= 1-1/n$$
for all $s \in \Omega$, and hence $I(\sigma_W)= 1-1/n$. Thus, by Theorem \ref{prop-best-fit},
\begin{equation}\label{eq-J}
J(\sigma_W)= \frac{n-2+\gamma}{n}\;.
\end{equation}

For instance, consider the (pure!)  Berezin-Toeplitz POVM $W_p$
from Example \ref{exam-cloudK}. Let us use formula \eqref{eq-J} in order to calculate $J$.
Recall that by the Riemann-Roch theorem (see \cite{Fine-lectures}, Propositions 2.25 and 4.21)
$$n_p = Vp^d + Up^{d-1}+ \bigo(p^{d-2})\;,$$
where $$V= \Vol(X) = [\omega]^d/d!,\;\;\;U = c_1(X) \cup [\omega]^{d-1}/(d-1)!\;.$$
It follows from formula \eqref{eq-gap-repeated} for $\gamma_p$ that
$$J(\sigma_{W_p})= 1  -\frac{2}{V}p^{-d} + \frac{8\pi U+ V\lambda_1}{4\pi V^2} p^{-d-1} +\bigo(p^{-d-2})\;.$$

For instance, for the dual to the tautological bundle over $\C P^1$ in Example \ref{exam-CP1}
$n=p+1$ and $\gamma= 2/(p+2)$ so by \eqref{eq-J}
$J = 1- \frac{2}{p+2}$.
}
\end{exam}

\medskip

Furthermore, we explore robustness of the gap $\gamma(W)$, as a function of the measure $\sigma_W$,
with respect to perturbations in the Wasserstein distances on the space of Borel
probability measures on $\cS(\Hilb)$. They are defined as follows. For compactly supported Borel
probability measures $\sigma_1,\sigma_2 $ on a metric space $(X,d)$ the $L_2$-Wasserstein distance is given by
$$\delta_2 \left( \sigma_1, \sigma_2 \right) := \inf \limits_{\nu} \left(\,\int \limits_{X\times X} \text{dist}\left(x_1, x_2\right)^2 \, d\nu(x_1,x_2)\right)^{1/2}\;,$$
and the $L_\infty$-Wasserstein distance by
$$\delta_\infty \left( \sigma_1, \sigma_2 \right) := \inf \limits_{\nu} \sup_{(x_1,x_2)\in \supp (\nu)} \text{dist}(x_1,x_2)\;,$$
where in both cases the infimum is taken over all Borel probability measures $\nu$ on $X \times X $ with marginals $\sigma_1$ and $\sigma_2$.

\begin{thm}
\label{thm-robust} Let $\sigma_V$ and $\sigma_W$  be measures on $\cS(\Hilb)$ associated
to POVMs $V$ and $W$ respectively.
\begin{itemize}
\item[{(i)}] $| \gamma(V) - \gamma(W)| \leq c(n) \delta_2(\sigma_V,\sigma_W)$, where $c(n)$ depends on the dimension $n=\dim \Hilb$;
\item[{(ii)}] If in addition $V$ and $W$ are pure POVMs, there exists a universal constant $c$ such that
\begin{equation}\label{eq-rob-pure}
|\gamma(V) - \gamma(W)| \leq c\delta_\infty(\sigma_V,\sigma_W)\;.
\end{equation}
\end{itemize}
\end{thm}

\medskip
\noindent Note that this result enables us to compare spectral gaps of POVMs defined
on different sets (but having values in the same Hilbert space). This idea goes back to \cite{OC} \footnote{In \cite{OC} the authors consider the $L_1$-version of this distance, and call it the Kantorovich distance.}. Let us emphasize that the estimate in (ii) is {\it dimension-free}. This is important, for instance, for comparison of spectral gaps
corresponding to different Berezin-Toeplitz quantization schemes.

Theorem \ref{thm-robust}(i) immediately follows
from the fact that $C(\sigma)$, $I(\sigma)$ and $J(\sigma)$ are Lipschitz in $\sigma$ with respect to $L_2$-Wasserstein distance. The details will appear in MSc thesis by V.~Kaminker.

For the proof of part (ii), we need the following auxiliary statement. In what follows we write $\|A\|_2$ for the Hilbert-Schmidt norm
$(\tr(AA^*))^{1/2}$.

\begin{lemma}\label{lemma-tr} Let $P,Q$ be rank $1$ orthogonal projectors.
Then for every $A \in \cL(\cH)$,
$$|\tr (A(P-Q))| \leq \sqrt{2}\|P-Q\|_2 \ (\tr (A^2(P+Q)))^{1/2}\;.$$
\end{lemma}

\begin{proof} Suppose that $P$ and $Q$ are orthogonal projectors to
unit vectors $\xi$ and $\eta$, respectively. By tuning the phase of $\xi$, we can assume that
$\langle \xi, \eta \rangle \geq 0$. We have
$$ |\tr (A(P-Q))| = |\langle A\xi,\xi\rangle - \langle A\eta,\eta\rangle| $$
$$= |\langle \xi-\eta, A\xi \rangle + \langle A\eta, \xi-\eta \rangle |\leq |\xi-\eta| (|A\xi|+|A\eta|)$$
$$ = |\xi-\eta|( \langle A^2\xi,\xi\rangle^{1/2} + \langle A^2\eta,\eta\rangle^{1/2}) \leq \sqrt{2}|\xi-\eta|( \langle A^2\xi,\xi\rangle + \langle A^2\eta,\eta\rangle)^{1/2}$$
$$ =\sqrt{2}|\xi-\eta|\left(\tr(A^2P)+\tr(A^2Q)\right)^{1/2} \;.$$
But since $0 \leq \langle \xi, \eta \rangle \leq 1$,
$$|\xi-\eta| = (2 -2\langle \xi, \eta \rangle)^{1/2}  \leq (2 -2\langle \xi, \eta \rangle^2)^{1/2}$$
$$= (\tr(P-Q)^2)^{1/2} = \|P-Q\|_2\;.$$ This completes the proof.
\end{proof}

\medskip
\noindent
{\bf Proof of Theorem \ref{thm-robust} (ii):}
Denote by $\cP$ the space of all rank 1 orthogonal projectors on $\Hilb$.
We can assume without loss of generality that pure POVMs $V$ and $W$
are defined on subsets $\Omega_1$ and $\Omega_2$ of $\cP$, respectively,
and that the maps $F_i: \Omega_i \to \cP$ are the inclusions.
Thus representation \eqref{eq-POVM-density} in this case can be simplified
as
$$dV(s) = n\,s\,d\alpha_1(s)\;, \; dW(t) = n\,t\,d\alpha_2(t)\;,$$
where $\sigma_V= \alpha_1$ and $\sigma_W=\alpha_2$ are Borel probability measures
supported in $\Omega_1$ and $\Omega_2$, respectively. Let us emphasize that here and below $s,t$ stand for rank 1 orthogonal projectors.
 Pick
any measure $\nu$ on $\cP \times \cP$ with marginals $\alpha_1$ and $\alpha_2$
and write
$$\Delta:= \max_{(s,t) \in \supp(\nu) } \|s-t\|_2\;.$$

We use the fact that the operators $\cE_1,\cE_2: \cL(\Hilb) \to \cL(\Hilb)$
given by formula \eqref{eq-e} have the same spectrum as the Berezin transform. For $A \in \cL(H)$ with $\tr (A^2) =1$ put
$$D:= |((\cE_1 A,A)) - ((\cE_2 A,A))|\;.$$ One readily rewrites
$$D = n \ \left | \int_{\Omega_1} ((F_1(s),A))^2d\alpha_1(s) - \int_{\Omega_2} ((F_2(t),A))^2d\alpha_2(t)\right | $$
$$\leq  n \ \int_{\Omega_1 \times \Omega_2} |\tr((s-t)A)|\ |\tr((s+t)A)| d\nu\;.$$
By Lemma \ref{lemma-tr}, $$|\tr((s-t)A)| \leq \sqrt{2}\|s-t\|_2\  (\tr(A^2(s+t)))^{1/2}\;.$$
By Cauchy-Schwarz, writing $$(s+t)A = (s+t)^{1/2}((s+t)^{1/2}A)\;,$$ we get
$$|\tr((s+t)A)| \leq (\tr(s+t))^{1/2}(\tr(A^2(s+t)))^{1/2} = \sqrt{2}(\tr(A^2(s+t)))^{1/2} \;.$$
It follows that
$$D \leq 2n\max_{(s,t) \in \supp(\nu)} \|s-t\|_2 \ \int \tr(A^2(s+t)) d\nu\;.$$
The integral on the right can be rewritten as
$$\tr \left( A^2 \int_{\Omega_1} s d\alpha_1(s)\right ) + \tr\left ( A^2 \int_{\Omega_2} t d\alpha_2(t)\right )= 2/n\;,$$
since  $$\int_{\Omega_1} s\,d\alpha_1(s) = \int_{\Omega_2} t\,d\alpha_2(t) = \frac{1}{n}\id$$ and $\tr(A^2)=1$.
It follow that $D \leq 4\Delta$. Choosing $\nu$ so that $\Delta$ becomes arbitrary close to
$\delta:= \delta_\infty(\alpha_1,\alpha_2)$, and taking $A$ with
\begin{equation} \label{eq-A-constraints}
\tr (A) =0\;,\ \tr(A^2)=1
\end{equation}
to be an eigenvector
of $\cE_1$ with the first eigenvalue $\gamma_1(\cE_1)$, we get that
$$|\gamma_1(\cE_1)- ((\cE_2 A,A))| \leq 4\delta\;.$$ But due to the variational characterization of the first
eigenvalue,  $\gamma_1(\cE_2) = \max ((\cE_2 A,A))$, where the maximum is taken over all $A$ satisfying  \eqref{eq-A-constraints}.
It follows that
$\gamma_1(\cE_1)- \gamma_1(\cE_2) \leq 4\delta$. By symmetry,
$\gamma_1(\cE_2)- \gamma_1(\cE_1) \leq 4\delta$, which yields the theorem with $c= 4$.

\qed

\medskip

Our next result provides a geometric characterization of the eigenfunction of the operator $\cB$ with the eigenvalue $\gamma_1$.
Let $A \in \cL(\Hilb)$ be the trace zero unit vector generating the best fitting line corresponding to $W$. In view of Theorem \ref{prop-best-fit},
$$\gamma_1 = 1-\gamma(W)= n(I-J)\;,$$ with $I= I(\sigma_W)$ and $J=J(\sigma_W)$.

\medskip
\noindent
\begin{thm}\label{thm-eigenf} The function \begin{equation}\label{eq-psi1}
\psi_1: \Omega \to \R,\; s \mapsto \frac{((F(s),A))}{\sqrt{I-J}}\end{equation}
is an eigenfunction of the operator $\cB$ with the eigenvalue $\gamma_1$. Furthermore,
$\|\psi_1\|=1$.
\end{thm}

\medskip
\noindent
In other words, up to a multiplicative constant, the first eigenfunction sends $s \in \Omega$
to the projection of the density $F(s)$ to the best fitting line.

\begin{proof}  By Remark \ref{rem-eigenv} above, the operator $A$ generating the
best fitting line is an eigenvector of the quantum channel $\cE$: $\cE A= \gamma_1 A$. Since $\cE=n^{-1}TT^*$ and $\cB= n^{-1}T^*T$, we have $\cB(T^*A)= \gamma_1 T^*A$ and $(T^*A,T^*A) = n\gamma_1$. Furthermore, $T^*A(s)= n((F(s),A))$ and $n\gamma_1 = n^2(I-J)$.
Choosing $\psi_1= T^*A/\|T^*A\|$, we get \eqref{eq-psi1}.
\end{proof}

\medskip

Next, we discuss the {\it diffusion distance} on $\Omega$ associated to the Markov operator $\cB$ (see \cite{CoL}). This distance, which originated in geometric analysis of data sets,  depends on a
positive parameter $\tau$ playing the role of the time in the corresponding random process. Take any orthonormal eigenbasis $\{\psi_k\}$ corresponding to eigenvalues $1=\gamma_0 \geq \gamma_1 \geq \gamma_2 \dots$ of $\cB$ such that
$\psi_0$ is constant. The diffusion distance $D_\tau$ is defined by
\begin{equation}\label{eq-diffusion}
D_\tau(s,t) = \big( \sum_{k \geq 1} \gamma_k ^{2\tau}(\psi_k(s) - \psi_k(t))^2 \big)^{1/2} \;\;\forall s,t \in \Omega\;.
\end{equation}
If $\gamma_1 < 1$, i.e., the spectral gap is positive, this expression decays exponentially.

Suppose now that $\gamma_2 < \gamma_1$. In this case the asymptotic behavior of $D_\tau(s,t)$ as $\tau \to \infty$  is given by
\begin{equation}
D_\tau(s,t)= \gamma_1^\tau\frac{|((F(s)-F(t),A))|}{(I-J)^{1/2}}\ (1+o(1))\;,\;\text{if}\;\; ((F(s),A))\neq ((F(t),A))\;,
\end{equation}
and $D_\tau(s,t) = \mathcal{O}(\gamma_2^\tau)$ otherwise.
The difference in these asymptotic formulas highlights  the multi-scale behaviour of the  metric space $(\Omega, D_\tau)$. In the first approximation, this space consists of the level sets of the function
$s \mapsto ((F(s),A))$ situated at the distance $\sim \gamma_1^\tau$ from one another, while each fiber
has the diameter $\lesssim \gamma_2^\tau$.  Viewing POVMs as data clouds in $\cS$ opens up a prospect of using various tools of geometric data analysis for studying POVMs. The above result on the diffusion distance associated to a POVM can be considered as a step in this direction.

\section{Case study: representations of finite groups} \label{sec-groups}
In this section we will be interested in finite POVMs associated to irreducible representations
of finite groups. We start with some preliminaries from Woldron's book \cite{W}. Let $G$ be a finite set.

\begin{defin}\label{def - tight frame} {\rm
	A finite collection $\{f_s\}_{s\in G}$ of non-zero vectors in a finite-dimensional Hilbert space $\Hilb$ is said to be a {\it tight frame} if there exists a number $A > 0$, called the {\it frame bound}, such that
	\begin{equation}\label{Bessel}
	A\| f \|^2 =\sum_{s\in G}|\langle f , f_s \rangle|^2, \forall f\in \Hilb\;.
	\end{equation}
}
\end{defin}
Denote by $P_s$ the orthogonal projector to $f_s$. One readily checks that for such a frame, the operators
\begin{equation}\label{eq-W-fr}
W_s:= \frac{\|f_s\|^2}{A} P_s, \; s \in G\;,
\end{equation}
form a $\cL(\Hilb)$-valued POVM on $G$.

Suppose from now on that $G$ is a finite {\it group}, and we are given its non-trivial irreducible unitary representation $\rho$
on a $d_\rho$-dimensional Hilbert space $V$. \footnote{All the representations considered below are assumed to be unitary.}  One can show \cite{W} that the vectors $\{f_s:=\frac{1}{\sqrt{d_\rho}}\rho(s)\}_{s\in G}$ form a tight frame in the operator space $\Hilb:=\text{End}(V)$
equipped with the Hermtian product $((C,D)) = \tr (CD^*)$ with the frame bound
$A={| G |}/{d_\rho ^2}$. Write $n=d_\rho^2 =\dim \Hilb$. By \eqref{eq-W-fr}, the corresponding POVM $W= \{W_s\}$, $s \in G$ is given by $W_s = n P_s \alpha_s$ with $\alpha_s = \frac{1}{| G |}$. Interestingly enough,
the spectrum of the corresponding Berezin transform can be calculated via the characters of irreducible unitary representations of $G$.

Denote by $\chi_\rho:G \to \C$, $\chi_\rho(s) := \tr (\rho(s))$ the character of the representation $\rho$.
Consider a basis in $L_2(G)$ consisting of the indicator functions of the elements of $G$.
It readily follows from the definition that the Berezin transform $\cB$ corresponding to the POVM $W$ is
given by a matrix
$$\cB_{ts} = n \,\tr(P_tP_s) \alpha_s = \frac{1}{| G |}u(st^{-1})\;,$$
where $u(s) := | \chi_\rho(s) | ^2$. The eigenvalues of this matrix and their multiplicities
are given by the following proposition, see chapter 3E of \cite{Diaconis}.
	\begin{prop} \label{lambda varphi}
		The eigenvalues of $\cB$ are given by  $$\lambda_\varphi := \frac{1}{d_\varphi| G|}\sum_{s\in G} u(s)\chi_\varphi(s)\;,$$ where $\varphi$ runs over irreducible representations of $G$, and the contribution of each $\varphi$ into the multiplicity of $\lambda_\varphi$ is $d^2_\varphi$.
	\end{prop}

\medskip
\noindent
Let us emphasize that it could happen that $\lambda_\varphi=\lambda_\psi$ for different representations $\varphi$ and $\psi$. Note also that by Lemma \ref{Useful properties}(i) below, $\lambda_\varphi = 1$ when $\varphi$ is the trivial one-dimensional representation.

\begin{rem}{\rm  We claim that the gap $\gamma(W)$ is rational. Indeed, for a unitary representation $\psi$
by complex unitary matrices denote $\psi' (s) = \overline{\psi(s)}$, where the bar stands for the complex conjugation.
Note that $u$ is the character
of the (in general, reducible) representation $\theta:= \rho \otimes \rho'$. By the Schur orthogonality relations, $$\frac{1}{| G|}\sum_{s\in G} u(s)\chi_\varphi(s)$$ equals the multiplicity of $\varphi$ in the decomposition
of $\theta$ into irreducible representations, and hence is an integer. The claim follows from Proposition \ref{lambda varphi}.
}
\end{rem}

The main result of this section is the following algebraic criterion of the positivity
of the spectral gap of $W$. Following chapter $12$ in \cite{Isaacs} we define the {\it vanishing-off subgroup} $\cV(\rho)$ to be the smallest subgroup of $G$ such that $\chi_\rho$ vanishes on $G \setminus \cV(\rho)$:
	$$\cV(\rho) = <s\in G \mid \chi_\rho(s)\not=0>.$$ Since the character $\chi_\rho$ is conjugation invariant, $\cV(\rho)$ is normal.\\
	
	\begin{thm}\label{thmgap0}
		The following are equivalent:
		\begin{itemize}
			\item[{(i)}] $\cV(\rho)\not=G$;
			 			\item[{(ii)}] $\gamma(W) = 0$, i.e., there exists a non-trivial irreducible
unitary representation $\varphi$ of $G$ with $\lambda_\varphi =1$.
		\end{itemize}
	\end{thm}

\medskip
In the next lemma, we collect some standard facts from the representation theory (see e.g. Chapter 2 of \cite{Diaconis}) which will be used in the proof of Theorem \ref{thmgap0}. We write $\text{Irrep}$ for the
set of all unitary irreducible representations of $G$ up to an isomorphism.

	\begin{lemma}\label{Useful properties}
			$(i)\ \frac{1}{| G |}\sum_{s\in G} | \chi_\rho(s) |^2 = 1,~\forall\,\rho \in \textup{Irrep}$;\\
			\\
$~~~~\quad\quad\quad\quad\quad(ii)\ \sum_{\varphi\in \textup{Irrep}} d_\varphi^2 = | G |.$
	\end{lemma}

\medskip\noindent{\bf Proof of Theorem \ref{thmgap0}:}

We begin by proving (ii) $\Rightarrow$ (i): Assume that there exists a non-trivial irreducible representation $\varphi$ with $\lambda_\varphi = 1$. By using Lemma \ref{Useful properties} and the explicit formula for the eigenvalues from Proposition \ref{lambda varphi} we see that
		\begin{equation*}
		\begin{split}
		1 & \underset{\ref{lambda varphi}}{=} \frac{1}{d_\varphi| G|} \sum_{s\in G} | \chi_\rho(s)|^2\chi_\varphi(s) = \\
		& = \frac{1}{| G|} \sum_{s\in G} | \chi_\rho(s)|^2 - \frac{1}{d_\varphi| G|} \sum_{s\in G} | \chi_\rho(s)|^2 (d_\varphi - \chi_\varphi(s)) \underset{\ref{Useful properties}}{=} \\
		& = 1 - \frac{1}{d_\varphi| G|} \sum_{s\in G} | \chi_\rho(s)|^2 (d_\varphi - \chi_\varphi(s))
		\end{split}
		\end{equation*}
		By taking the real part of both sides we get
		$$\sum_{s\in G} | \chi_\rho(s) |^2\,\text{Re}(d_\varphi - \chi_\varphi(s)) = 0$$
		 Note that since  $\varphi(s)$ is unitary, all its eigenvalues are of the form $e^{i\theta}$ so $| \chi_\varphi (s) | = | \tr(\varphi(s)) |\leq d_\varphi$ and $\chi_\varphi(s)=d_\varphi$ iff $\varphi(s)$ is the identity. Hence $\text{Re}(d_\varphi - \chi_\varphi(s))$ must be non-negative. Since $| \chi_\rho(s) |^2$ is also non-negative, $\chi_\varphi(s)=d_\varphi$ for every $s \in G$ with $\chi_\rho(s) \neq 0$.
As we have seen above $\chi_\varphi(s)=d_\varphi$ if and only if $\varphi(s)=\id$. It follows that the vanishing off subgroup $\cV(\rho)$ is contained in the normal subgroup
$$\text{Ker}(\varphi) := \{s\mid\varphi(s)=\id\}\;.$$ Since $\varphi$ is irreducible and non-trivial, the latter subgroup
$\neq G$, and hence $\cV(\rho) \neq G$, as required.

\medskip

Next, we prove (i) $\Rightarrow$ (ii): Assume $\cV(\rho)\not=G$. Consider the quotient $H:=G/\cV(\rho)$, which is a non-trivial group, and let $\pi: G \to H$ be the natural projection. Take any non-trivial irreducible representation $\psi$ of $H$. Then $\varphi:= \psi \circ \pi$
is an irreducible representation of $G$. We claim that $\lambda_\varphi=1$. Indeed, for $s \in \cV(\rho)$ we have $\varphi(s)=\id$ and hence $\chi_\varphi(s) = d_\varphi$, and  for $s \notin \cV(\rho)$ holds $\chi_\rho(s) = 0$. It follows that

		\begin{equation*}
		\begin{split}
		\lambda_{\varphi} = \sum_{s\in \cV(\rho)} \frac{1}{d_{\varphi} | G |} | \chi_\rho(s) |^2 d_{\varphi} = \frac{1}{| G |} \sum_{s\in G} | \chi_\rho(s) |^2 \underset{\ref{Useful properties}}{=} 1\;.
		\end{split}
		\end{equation*}
This proves the claim and hence completes the proof of the theorem.

\qed

\medskip

\begin{cor}\label{cor-simplegr}
		If $G$ is a simple group, then the gap of $W$ is positive.
	\end{cor}
\begin{proof} Indeed, otherwise by Theorem \ref{thmgap0} and the simplicity of $G$,
$\cV(\rho) = \{\id\}$, which means that $\chi_\rho(s) =0$ for every $s \neq \id$. Then the first
statement of Lemma \ref{Useful properties} yields $|G| = d_\rho^2$, while the second statement
guarantees that $|G| \geq 1+ d_\rho^2$, since $\rho$ is a non-trivial representation. We get a contradiction.
\end{proof}

\medskip

Let us  point out that there exist non-simple groups $G$ admitting an irreducible representation $\rho$
with $\cV(\rho)=G$. Indeed, consider the irreducible representation $\rho:\mathbb{Z}_m\to U(\mathbb{C}),\ \rho(s)=e^{2\pi i s/m}$ of  the abelian cyclic group $\mathbb{Z}_m$. Observe that $\cV(\rho) = \Z_m$, while $\Z_m$ is simple
if and only if $m$ is prime.

\medskip

Let us describe the diffusion distance $D_\tau$ (see \eqref{eq-diffusion}) corresponding to the POVM $W$ associated
to a finite group $G$ and a non-trivial irreducible representation $\rho$. Recall \cite{Diaconis} that for an irreducible representation $\varphi: G \to \mathbb{U}(n)$, the orthonormal basis of eigenfunctions corresponding to the eigenvalue $\lambda_\varphi$ presented in Proposition \ref{lambda varphi} is given by the matrix coefficients of $\varphi$ multiplied by $\sqrt{d_\varphi}$.  Assume that the gap of $G$ is strictly positve, and denote by $\beta_1 > \dots > \beta_k$ all {\it pair-wise distinct} eigenvalues of $\cB$ lying in the open interval $(0,1)$. Denote
$$R_j := \{\varphi \in \textup{Irrep}\mid \lambda_{\varphi} = \beta_j\}\;.$$
Then \eqref{eq-diffusion} yields the following expression for the diffusion distance:
\begin{equation}\label{eq-diffusion-1}
D_\tau(s,t) = \left ( \sum_{j=1^k} \beta_j^{2\tau} \sum_{\varphi \in R_j} d_\varphi \|\varphi(s)-\varphi(t)\|_2^2 \right )^{1/2} \;,
\end{equation}
where $\|\;\|_2$ stands for the Hilbert-Schmidt norm $\|C\|_2= (\tr (CC^*) )^{1/2}$.
Note that this expression can be rewritten in terms of the character $\chi_\varphi$ since
$$\|\varphi(s)-\varphi(t)\|_2^2  = 2(d_\varphi- \text{Re}\;\chi_\varphi(st^{-1}))\;.$$

Define a normal subgroup
$\Gamma_j := \bigcap_{\varphi \in R_j}  \text{Ker} (\varphi)\;, j= 1,\dots, k$
and a normal series $K_0 \supset K_1 \supset \dots ...$ with $K_0 = G$, $K_{k+1} =\{1\}$ and
$$K_m : = \bigcap_{j=1}^m \Gamma_j\;, m=1,\dots k\;.$$ It follows from
\eqref{eq-diffusion-1} that for $\tau \to +\infty$
\begin{equation}\label{eq-diffusion-2}
D_\tau(s,t) \sim \beta_{p+1}^\tau \;\;\text{for}\;\; st^{-1} \in K_p \setminus K_{p+1}\;.
\end{equation}
In fact we have a sequence of nested partitions $\Delta_p$ of $G$ formed by the cosets of $K_p$.
For every pair of distinct points $s,t \in G$ choose maximal $p$ so that $s$ and $t$ lie in the same element
of $\Delta_p$. Then asymptotical formula \eqref{eq-diffusion-2} holds, which manifests the multi-scale nature
of the diffusion distance.

Let us illustrate this in the case when $G=S_4$ is the symmetric group, and $\rho$ a $3$-dimensional
irreducible representation. The direct calculation with the character table of $S_4$ shows that
the first non-trivial eigenvalue $1/2$ corresponds to the unique $2$-dimensional irreducible representation
whose kernel coincides with the normal subgroup $K$ of order $4$ of $S_4$ called the Klein four-group. Thus $D_\tau(s,t) \sim (1/2)^\tau$ if $s,t$ belong to different cosets of $K$ in $S_4$,
and one can calculate that $D_\tau(s,t) \sim (1/3)^\tau$ if $s,t$ are distinct and belong to the same coset.

\begin{rem}\label{rem-rps}{\rm
A modification of the construction presented in this section is related
to Berezin-Toeplitz quantization. The modification goes in two directions.
First, we deal with unitary representations $\rho$ of compact Lie groups $G$ instead of finite groups, and second,
our POVMs are related to the $G$--orbits in a representation space $\Hilb$
as opposed to the image of $\rho$ in the endomorphisms of $\Hilb$. Let us very briefly
illustrate this in the following simplest case. Consider the irreducible unitary representation $\rho_j$ of the group $G=SU(2)$
in an $n= 2j+1$-dimensional Hilbert space $\Hilb$, $j \in \frac{1}{2}\N$.  Fix a maximal torus $K=S^1 \subset G$, and let $w \in \Hilb$
be the maximal weight vector of $K$, that is $\rho_j(t)w= e^{4\pi i j t}w$ for all $t \in K$. Consider an  $\cL(\Hilb)$-valued POVM
$W$ on $\Omega=G/K=\C P^1$ of the form $dW ([g]) = nP_{[g]}d\alpha([g])$, where $[g]$ stands for the class of $g \in G$ in $\Omega$, $\alpha$ is the $G$-invariant measure on $\Omega$ and $P_{[g]}$ is the rank one projector to $gw$. Note that $W$ is
nothing else but the Berezin-Toeplitz POVM $W_{p}$ from Example \ref{exam-CP1} with $p=2j$. We refer to \cite[Chapter 7]{CR} for the representation theoretic approach to coherent states and quantization. By using theory
of Gelfand pairs (cf. \cite[Chapter 3.F]{Diaconis}) one can check that the eigenvalues of the Berezin transform
are of the form $\lambda_\varphi = (u,\chi_\varphi)_{L_2}$ , where $\varphi$ runs over all irreducible unitary representations of $G$,
$\chi_\varphi$ stands for the character of $\varphi$ and $u(g)=n|\langle \rho(g) w, w\rangle|^2$. The multiplicity of $\lambda_\varphi$ equals $d_\varphi$, where $d_\varphi$ is the dimension of $\varphi$. In order to calculate $\lambda_\varphi$, recall that
\begin{equation}\label{eq-square}
\rho_j \otimes \rho_j = \bigoplus_{k=0}^{2j} \rho_k\;.
\end{equation}
Writing $v$ for the vector of weight $-j$ of $\rho_j$,
we have
$$u(g)=n \langle (\rho_j \otimes \rho_j) (g)\xi,\xi\rangle,\;\,\text{where}\;\; \xi = w\otimes v\;.$$
In order to complete this calculation, one has to decompose $\xi$ in the sense of \eqref{eq-square}.
This can be done with the help of explicit expressions for the Clebsch-Gordan coefficients, and it eventually
yields eigenvalues of the Berezin transform, including $\gamma_1 = j/(j+1)$ (cf. Example \ref{exam-CP1}),
in agreement with calculations by Zhang \cite{gap-representations} and Donaldson
 \cite[p.\,613]{D}. The details
will appear in MSc thesis by D. Shmoish.}
\end{rem}

\section{Two concepts of quantum noise}\label{sec-noise}
In the present section we provide two different (and essentially tautological)
interpretations of the spectral gap in the context of quantum noise. In quantum measurement theory, there are two concepts of quantum noise: the increment of variance for unbiased approximate measurements as formalized by the noise operator, see below, and a non-unitary evolution of a quantum system described by a quantum channel (a.k.a. a quantum operation, see, e.g. \cite[Chapter 8]{NC}). Such a non-unitary evolution can be caused, for instance, by the quantum state reduction in the process of repeated quantum measurements. Interestingly enough, for pure POVMs, the spectral gap $\gamma(W)$ brings together these two seemingly remote concepts: it measures the minimal magnitude of noise production in the context of the noise operator, and it equals the spectral gap of the Markov chain modeling repeated quantum measurements.

Given an observable $A \in \cL(\Hilb)$, write $A= \sum \lambda_iP_i$ for its spectral decomposition, where $P_i$'s are pair-wise distinct orthogonal projectors. According to the statistical postulate of quantum mechanics, in a state $\rho$ the observable $A$ attains value $\lambda_i$ with probability $((P_i,\rho))$.
It follows that the expectation of $A$ in $\rho$ equals $\mathbb{E}(A,\rho)= ((A,\rho))$ and the variance
is given by $\mathbb{V}ar(A,\rho) = ((A^2,\rho)) - \mathbb{E}(A,\rho)^2$. In quantum measurement theory \cite{Busch-2}, a POVM $W$  represents a measuring device coupled with the system, while $\Omega$ is interpreted as the space of device readings. When the system is in a  state $\rho \in \cS(\Hilb)$, the probability of finding the device in a subset $X \in \cC$ equals $\mu_{\rho}(X):= ((W(X),\rho))$.
An experimentalist performs a measurement whose outcome, at every state $\rho$, is distributed in $\Omega$ according to the measure $\mu_{\rho}$. Given a function $\phi \in L_2(\Omega,\alpha)$ (experimentalist's choice), this procedure yields {\it an unbiased approximate measurement} of the quantum observable $A:=T(\phi)$. The expectation
of $A$ in every state $\rho$ equals $((A,\rho))$ and thus coincides with the one of the measurement procedure
given by $\int_\Omega \phi d\mu_\rho$ (hence {\it unbiased}), in spite of the fact that actual probability distributions determined by the observable $A$ (see above) and the random variable $(\phi,\mu_{\rho})$  could be quite different (hence {\it approximate}). In particular, in general, the variance increases under an unbiased approximate measurement:
\begin{equation}\label{eq-noise}
\mathbb{V}ar(\phi,\mu_{\rho}) = \mathbb{V}ar(A,\rho) + ((\Delta_W(\phi),\rho))\;,
\end{equation}
where
$
\Delta_W(\phi):= T(\phi^2)-T(\phi)^2
$
is {\it the noise operator}. This operator, which is known to be positive, measures the increment of the variance. We wish to explore the relative magnitude of this increment for the ``maximally mixed" state
$\theta_0= \frac{1}{n}\id$. To this end introduce {\it the minimal noise} of the POVM $W$ as $$\cN_{min}(W):= \inf_{\phi} \frac{((\Delta_W(\phi),\theta_0))}{\mathbb{V}ar(\phi,\mu_{\theta_0}) }\;,$$ where the infimum is taken over all non-constant functions $\phi \in L_2(\Omega,\alpha)$. It turns out that the minimal noise coincides with the spectral gap:
\begin{equation}\label{eq-Nmingamma}
\cN_{min}(W) = \gamma(W)\;.
\end{equation}
Indeed, since $\tr(T(\phi^2))= n(\phi,\phi)$, we readily get that
$$((\Delta_W(\phi),\theta_0))= ((\id- \cB)\phi,\phi)\;,$$
where $\cB= n^{-1}T^*T$ is the Markov operator given by \eqref{eq-b},
while $$\mathbb{V}ar(\phi,\mu_{\theta_0})= (\phi,\phi) - (\phi,1)^2\;.$$
Formula \eqref{eq-Nmingamma} follows from the variational principle.

\medskip

Suppose now that $\Omega \subset \cS(\Hilb)$ is a finite set consisting of rank one projectors $\{P_1,\dots, P_N\}$ and that $W$ is a pure POVM of the form $W(P_i) := n\alpha_i P_i$, where $\alpha$ is a probability measure on $\Omega$. Given a system in the original state $\rho$, the result of the measurement equals $P_j$ with probability $p= n\alpha_j((P_j,\rho))$. Recall the quantum state reduction (a.k.a. the wave function collapse) axiom for so called {\it L\"{u}ders} repeated quantum measurements: if the result of the measurement equals $P_j$, the system moves from the original state $\rho$ to the new (reduced) state $$\rho' = \frac{1}{p} W(P_j)^{1/2}\rho W(P_j)^{1/2}= P_j\;.$$
It follows that if the original state $\rho$ is chosen from $\Omega$, the repeated quantum measurements are described by the Markov chain
with transition probabilities $n\alpha_j((P_i,P_j))$. The corresponding Markov operator equals $\cB$,
and {\it the spectral gap of the Markov chain coincides with the spectral gap $\gamma(W)$} of the POVM $W$.
Furthermore, given an original state $\rho \in \Omega$, the expected value of the reduced state equals $\cE(\rho)$.
It follows that if $\gamma(W)>0$,  $\cE^k(\rho)$, $k \to \infty$ converge to the maximally mixed quantum state $\frac{1}{n}\id$ at the exponential rate $\sim (1-\gamma(W))^k$. In other words, {\it for pure POVMs the spectral gap controls the convergence rate to the maximally mixed state under repeated quantum measurements}.

\medskip
\noindent
{\bf Acknowledgment.} We are grateful to D.~Aharonov,  S.~Artstein-Avidan, J.-M. Bismut, L.~Charles, G.~Kalai, G.~Kindler, M.~Krivelevich, S.~Nonnenmacher, Y.~Ostrover, A.~Reznikov, A.~Veselov and S.~Weinberger for useful discussions on various aspects of this work. We thank O.~Oreshkov for bringing our attention to the paper \cite{OC} and illuminating comments, as well as I.~Polterovich for very helpful remarks. We thank A.~Deleporte for sharing with us his ideas on the proof of Theorem \ref{thm-quant}, and J.~Fine for providing us highly useful insights and references on Donaldson's program. Finally, we thank the anonymous referee for numerous useful comments.

\medskip

\begin{tabular}{l}
School of Mathematical Sciences\\
Tel Aviv University\\
Israel\\
\end{tabular}
\end{document}